%% file: main.tex
\documentclass[a4paper,10pt]{article} 

\newcommand{\typeof}{0} %

\usepackage{ifthen}
\newcommand{\condinc}[2]{\ifthenelse{\equal{\typeof}{0}}{#1}{#2}}

\input{macros}
\input{newcommands}

\begin{document}

\maketitle
\begin{abstract}
  Linear dependent types~\cite{DLG11} allow to precisely capture both the extensional behaviour and the time complexity of $\lambda$-terms, 
  when the latter are evaluated by Krivine's abstract machine. In this work, we show that the same paradigm can be applied to 
  call-by-value evaluation. A system of linear dependent types for Plotkin's \PCF\ is introduced, called \dlpcfv, 
  whose types reflect the complexity of evaluating terms in the so-called \CEK\ machine. \dlpcfv\ is proved to be sound, but also relatively complete:
  every true statement about the extensional and intentional behaviour of terms can be derived, provided all true index term 
  inequalities can be used as assumptions.
\end{abstract}

\condinc{}{


}

\input{intro}
\input{intuition}
\input{formalism}
\input{examples}
\input{metatheory}
\input{dev}
\input{concl}

\bibliographystyle{abbrv}
\bibliography{bibli}

\end{document}

%% file: macros.tex
\condinc{
  \title{Linear Dependent Types\\ in a Call-by-Value Scenario}
  \author{Ugo Dal Lago \and Barbara Petit}
}{
  \title{Linear Dependent Types\\ in a Call-by-Value Scenario\titlenote{This work 
      is partially supported by the INRIA ARC Project ``ETERNAL''}}
  \authorinfo{Ugo Dal Lago \and Barbara Petit}{
    Universit\`a di Bologna \& INRIA
  }{ \{dallago,petit\}@cs.unibo.it}
  \conferenceinfo{PPDP'12,} {September 19--21, 2012, Leuven, Belgium.}
  \CopyrightYear{2012}
  \copyrightdata{978-1-4503-1522-7/12/09}
}

\usepackage[utf8x]{inputenc}
\condinc{\usepackage{a4wide}}{\usepackage{comment}}
\usepackage{hyperref}
\usepackage{graphicx}
\usepackage{proof,bussproofs}
\usepackage{xcolor}
\usepackage{amsmath,amssymb}
\usepackage{mathabx}
\usepackage{qsymbols}
\usepackage{mathrsfs}
\usepackage[all,color,frame]{xy}
\usepackage{theorem}

\newtheorem{theorem}{Theorem}[section]
\newtheorem{lemma}[theorem]{Lemma}
\newtheorem{proposition}[theorem]{Proposition}

\condinc{}{}
\condinc{
  \newenvironment{proof}{\begin{trivlist}
    \item[\hskip \labelsep {\bfseries Proof.}]}{
      \hfill $\Box$ \end{trivlist}}
}{\excludecomment{proof}}

\newenvironment{varitemize}
{
\begin{list}{\labelitemi}
{\setlength{\itemsep}{0pt}
 \setlength{\topsep}{0pt}
 \setlength{\parsep}{0pt}
 \setlength{\partopsep}{0pt}
 \setlength{\leftmargin}{15pt}
 \setlength{\rightmargin}{0pt}
 \setlength{\itemindent}{0pt}
 \setlength{\labelsep}{5pt}
 \setlength{\labelwidth}{10pt}
}}
{
 \end{list} 
}

\newcounter{numberone}

\newenvironment{varenumerate}
{
\begin{list}{\arabic{numberone}.}
{
  \usecounter{numberone}
  \setlength{\itemsep}{0pt}
  \setlength{\topsep}{0pt}
  \setlength{\parsep}{0pt}
  \setlength{\partopsep}{0pt}
  \setlength{\leftmargin}{15pt}
  \setlength{\rightmargin}{0pt}
  \setlength{\itemindent}{0pt}
  \setlength{\labelsep}{5pt}
  \setlength{\labelwidth}{15pt}
}}
{
\end{list} 
}


%% file: newcommands.tex
\newcommand{\dlpcfn}{\ensuremath{\mathsf{d}\ell\mathsf{PCF_N}}}
\newcommand{\dlpcfv}{\ensuremath{\mathsf{d}\ell\mathsf{PCF_V}}}
\newcommand{\BLL}{\textsf{BLL}}                      
\newcommand{\PCF}{\ensuremath{\mathsf{PCF}}}
\newcommand{\CEK}{\ensuremath{\mathsf{CEK}}}
\newcommand{\CEKPCF}{\ensuremath{\mathsf{CEK}_{\mathsf{PCF}}}}
\newcommand{\DML}{\ensuremath{\mathsf{DML}}}
\newcommand{\KAM}{\ensuremath{\mathsf{KAM}}}
\newcommand{\IL}{\ensuremath{\mathsf{IL}}}
\newcommand{\ILL}{\ensuremath{\mathsf{ILL}}}
\newcommand{\cek}{\ensuremath{\text{CEK}_{\mathsf{PCF}}}} 
\newcommand{\tov}{\ensuremath{→_v}}                  
\newcommand{\cbn}{\textsc{cbn}}                      
\newcommand{\cbv}{\textsc{cbv}}                      

\newcommand{\ie}{\textit{i.e.}}
\newcommand{\eg}{\textit{e.g.} }
\newcommand{\cf}{\textit{cf.} }

\renewcommand{\t}{\ensuremath{t}}          
\newcommand{\s}{\ensuremath{s}}            
\renewcommand{\u}{\ensuremath{u}}          
\renewcommand{\v}{\ensuremath{v}}          
\newcommand{\w}{\ensuremath{w}}            
\newcommand{\ifz}[3]{\ensuremath{
    \mathtt{\ ifz\ }#1\mathtt{\ then\ }#2\mathtt{\ else\ }#3}} 
\newcommand{\fix}[2][x]{\ensuremath{\mathtt{\ fix\ }#1.#2}}    
\newcommand{\nb}[1][n]{\ensuremath{\underline{\mathtt #1}}}    
\newcommand{\suc}{\ensuremath{\mathtt{s}}}                     
\newcommand{\pred}{\ensuremath{\mathtt{p}}}                    
\newcommand{\tsubst}[2][x]{\ensuremath{[#1:=#2]}}              
\newcommand{\ms}[1]{\|#1\|}   
\newcommand{\ts}[1]{|#1|}     
\newcommand{\sis}[1]{|#1|}    
\newcommand{\sic}[1]{|#1|}    
\newcommand{\sip}[1]{|#1|}    

\newcommand{\termdbl}{\ensuremath{\mathsf{dbl}}}        
\newcommand{\termdbltwo}{\ensuremath{\mathsf{dbl2}}}    
\newcommand{\termdiv}{\ensuremath{\mathsf{div}}}        
\newcommand{\termadd}{\ensuremath{\mathsf{add}}}        
\newcommand{\termmult}{\ensuremath{\mathsf{mult}}}        

\newcommand{\eval}[3]{#1\Downarrow^{#2}#3}

\newcommand{\one}{\ensuremath{\mathtt 1}}          
\newcommand{\zero}{\ensuremath{\mathtt 0}}         
\newcommand{\I}{\ensuremath{\mathrm{I}}}           
\newcommand{\J}{\ensuremath{\mathrm{J}}}           
\newcommand{\K}{\ensuremath{\mathrm{K}}}           
\renewcommand{\H}{\ensuremath{\mathrm{H}}}         
\renewcommand{\L}{\ensuremath{\mathrm{L}}}         
\newcommand{\M}{\ensuremath{\mathrm{M}}}           
\newcommand{\N}{\ensuremath{\mathrm{N}}}           
           % ...
\newcommand{\inb}[1][n]{\ensuremath{\mathrm{#1}}}  
\newcommand{\f}{\ensuremath{\mathrm{f}}}           
\newcommand{\g}{\ensuremath{\mathrm{g}}}           
           % ...
\newcommand{\mnu}{\ensuremath{-}}                  
\newcommand{\fc}[4][b]{\bigotriangleup_{#1}^{#2,#3}#4}        
\newcommand{\isubst}[2][a]{\ensuremath{\{#2/#1\}}} 

\newcommand{\Dfscomb}[4]{\bigotriangleup_{#1}^{#2,#3}#4}

\newcommand{\vars}{\mathcal{V}}
\newcommand{\cvg}{\Downarrow}

\newcommand{\A}{\ensuremath{A}}           
\newcommand{\B}{\ensuremath{B}}           
\newcommand{\C}{\ensuremath{C}}           
           % ...
\newcommand{\freccia}[2]{\ensuremath{#1\multimap#2}}  
\newcommand{\mtyp}[3][a]{\ensuremath{[#1<#2]\cdot#3}} 
\newcommand{\btyp}[3][a]{\ensuremath{!_{#1<#2}#3}}     
\newcommand{\Nat}[2]{\ensuremath{\mathtt{Nat}[#1,#2]}}
\newcommand{\NatU}[1]{\ensuremath{\mathtt{Nat}[#1]}}  
\newcommand{\NatPCF}{\ensuremath{\mathtt{Nat}}}       
\newcommand{\arr}{\Rightarrow} 
\newqsymbol{`<}{\sqsubseteq}              
\newqsymbol{`=}{\equiv}                   

\newcommand{\trad}[1]{\ensuremath{#1^{`o}}}   
\newcommand{\erase}[1]{(\!|#1|\!)}            

\newcommand{\ep}{\ensuremath{\mathcal{E}}}      
\newcommand{\fiv}{\ensuremath{`f}}              
\newcommand{\ictx}{\ensuremath{`F}}             
\newcommand{\ictxvar}{\ensuremath{\Psi}}        
\newcommand{\judg}[7][\ep]{\ensuremath{#2;#3;#4\vdash_{#5}^{#1}#6:#7}} 
\newcommand{\ejudg}[5][\ep]{\ensuremath{#2\vdash_{#3}^{#1}#4:#5}} 
\newcommand{\ijudg}[4][\ep]{\ensuremath{#2;#3\vDash_{#1}#4}}      
\newcommand{\eijudg}[2][\ep]{\ensuremath{\vDash_{#1}#2}}          
\newcommand{\sjudg}[4][\ep]{\ensuremath{#2;#3\vdash_{#1}#4}}      
\newcommand{\tder}{\ensuremath{`q}}             
\newcommand{\pcftder}{\ensuremath{`d}}          
\newcommand{\proves}{\ensuremath{\rhd}}         
\newqsymbol{`0}{\emptyset}                      
\newqsymbol{`|}{\Downarrow}                     

\newcommand{\env}[1][`x]{\ensuremath{#1}}       
\newcommand{\envelt}[2][x]{\ensuremath{(#1\mapsto#2)}} 
\newcommand{\clo}[2][\env]{\ensuremath{\langle\,#2\,;\,#1\,\rangle}} 
\newcommand{\cloc}[1][c]{\ensuremath{\mathbf{\sf #1}}} 
\newcommand{\clov}[1][v]{\ensuremath{\mathbf{\sf #1}}} 
\newcommand{\stack}[1][`p]{\ensuremath{#1}}            
\newcommand{\emptystack}{\ensuremath{\diamond}}
\newcommand{\stackfun}[2][\stack]{\ensuremath{\mathsf{fun}(#2\,,\,#1)}}
\newcommand{\stackarg}[2][\stack]{\ensuremath{\mathsf{arg}(#2\,,\,#1)}}
\newcommand{\stackif}[4][\stack]{\ensuremath{\mathsf{fork}(#2\,,\,#3\,,\,#4\,,\,#1)}}
\newcommand{\stackp}[1][\stack]{\ensuremath{\mathsf{p}(#1)}}
\newcommand{\stacks}[1][\stack]{\ensuremath{\mathsf{s}(#1)}}
\newcommand{\process}[1][P]{\textsc{#1}}
\newcommand{\R}{\textsc{R}}

\newcommand{\proc}[2]{\ensuremath{#1~\star~#2}}
\newcommand{\tocek}{\ensuremath{\succ}}
\newcommand{\conv}[1][n]{\ensuremath{\Downarrow^{#1}}}    

\newcommand{\fifi}{$(\fiv;\ictx)$}
\newcommand{\cjudg}[6][\ep]{\ensuremath{#2;#3\vdash_{#4}^{#1}#5:#6}} 
\newcommand{\pjudg}[6][\ep]{\ensuremath{#2;#3\vdash_{#4}^{#1}#5:#6}} 
\newcommand{\stjudg}[7][\ep]{\ensuremath{#2;#3\vdash_{#4}^{#1}#5:(#6,#7)}} 
\newcommand{\pcftjudg}[3]{\ensuremath{#1\vdash_{\PCF}#2:#3}} 
\newcommand{\pcfpjudg}[2]{\ensuremath{\vdash_{\PCF}#1:#2}} 
\newcommand{\pcfstjudg}[3]{\ensuremath{\vdash_{\PCF}#1:(#2,#3)}} 

\renewcommand{\int}{\ensuremath{\mathbb{N}}}          

\newcommand{\itp}[2][`r]{\ensuremath{\llbracket #2\rrbracket^{\ep}_{#1}}}  
\newqsymbol{`?}{\maltese}                             


%% file: intro.tex
\section{Introduction}\label{sec:intro}
A variety of methodologies for formally verifying properties of programs
have been introduced in the last fifty years. Among them, \emph{type systems}
have certain peculiarities. On the one hand, the way one defines
a type system makes the task of proving a given program to have a type
reasonably simple and modular: a type derivation for a compound program
usually consists of some type derivations for the components, appropriately
glued together in a syntax-directed way (i.e. attributing a type to a program 
can usually be done \emph{compositionally}). On the other, the
specifications that can be expressed through types have traditionally been 
weak, although stronger properties have recently become of interest, such as 
security~\cite{VolpanoIS96,SabelfeldMyers03}, 
termination~\cite{BartheGR08}, monadic temporal 
properties~\cite{KobayashiO09} or resource bounds~\cite{HOAA10}.
But contrarily to what happens with other formal methods (\eg model checking or program
logics), giving a type to a program~$\t$ is a \emph{sound} but \emph{incomplete} way to
prove~$\t$ to satisfy a specification: there are correct programs which cannot 
be proved such by way of typing.

In other words, the tension between expressiveness and tractability is particularly 
evident in the field of type systems, where certain good properties the majority
of type systems enjoy (\eg syntax-directedness) are usually
considered as desirable (if not necessary), but also have their drawbacks: some specifications
are intrinsically hard to verify locally and compositionally.
One specific research field in which the just-described scenario 
manifests itself is complexity analysis, in which the specification
takes the form of concrete or asymptotic bounds on the complexity
of the underlying program. Many type systems have been introduced
capturing, for instance, the class of polynomial time computable
functions~\cite{Hofmann99a,BaillotT09,BaillotGaboardiMogbil09esop}. 
All of them, under mild assumptions, can be employed
as tools to certify programs as asymptotically time efficient. 
However, a tiny slice of the polytime \emph{programs} are generally typable, since
the underlying complexity class $\mathbf{FP}$ is only characterised 
in a purely extensional sense --- for every function in $\mathbf{FP}$ 
there is \emph{at least one} typable program computing it.

Gaboardi and the first author have recently introduced~\cite{DLG11} 
a type system for Plotkin's \PCF, called \dlpcfn, in which linearity 
and a restricted form of dependency in the spirit of Xi's \DML\ are
present:
\begin{itemize}
\item
  \textbf{Linearity} makes it possible to finely control the number of times
  subterms are copied during the evaluation of a term~$\t$, 
  itself a parameter which accurately reflects the time complexity of~$\t$~\cite{DalLago09a}.
\item
  \textbf{Dependency} allows to type distinct (virtual) copies of a term with
  distinct types. This gives the type system an extra flexibility
  similar to that of intersection types.
\end{itemize}
When mixed together, these two ingredients allow to precisely capture the
extensional behaviour of $\lambda$-terms \emph{and} the time complexity of their
evaluation by Krivine's abstract machine. Both soundness and relative
completeness hold for \dlpcfn.

One may argue, however, that the practical relevance of these results is
quite limited, given that call-by-name evaluation and \KAM\ are very inefficient:
why would one be interested in verifying the complexity of evaluating concrete
programs in such a setting?

In this work, we show that linear dependent types can also
be applied to the analysis of call-by-value evaluation of functional programs. 
This is done by introducing another system of linear dependent types for Plotkin's \PCF. 
The system, called \dlpcfv, captures the complexity
of evaluating terms by Felleisen and Friedman's \CEK\ machine~\cite{FelleisenF87}, 
a simple abstract machine for call-by-value evaluation. \dlpcfv\ is proved 
to enjoy the same good properties enjoyed by its sibling \dlpcfn, namely soundness
and relative completeness: every true statement about the extensional
behaviour of terms can be derived, provided all true index term 
inequalities can be used as assumptions.

Actually, \dlpcfv\ is not merely a variation on \dlpcfn: not only typing
rules are different, but also the language of types itself must be modified.
Roughly, \dlpcfv\ and \dlpcfn\ can be thought as being
induced by translations of intuitionistic logic into linear logic: the latter
corresponds to Girard's translation $A\arr B\equiv !A\multimap B$,
while the former corresponds to $A\arr B\equiv !(A\multimap B)$.
The strong link between translations of \IL\ into \ILL\ and notions of reduction
for the $\lambda$-calculus is well-known (see \eg\cite{MaraistOTW95}) and has been a guide
in the design of \dlpcfv\ (this is explained in Section.~\ref{sec:cbvcbn}). 

\condinc{}{An extended version with more details and proofs is available~\cite{EV}.}


%% file: intuition.tex
\section{Linear Dependent Types,\\ Intuitively}
\label{sec:intui}
Consider the following program:
$$
\termdbl=\fix[f]{`lx.\ifz{x}{x}{\suc(\suc(f(\pred(x))))}}.
$$
In a type system like \PCF~\cite{plotkin77tcs}, the term~\termdbl\ receives type $\mathsf{Nat}\arr\mathsf{Nat}$. 
As a consequence, \termdbl\ computes a function on natural numbers without ``going wrong'':
it takes in input a natural number, and (possibly) produces in output another natural number.
The type $\mathsf{Nat}\arr\mathsf{Nat}$, however, does not give any information about 
\emph{which} specific function on the natural numbers~\termdbl\ computes. 

Properties of programs which are completely ignored by ordinary type systems are 
termination and its most natural refinement, namely termination in \emph{bounded time}.
Typing a term $\t$  with $\mathsf{Nat}\arr\mathsf{Nat}$ does not guarantee that $t$, when
applied to a natural number, terminates. Consider, as another example, a slight modification of $\termdbl$, namely
$$
\termdiv=\fix[f]{`lx.\ifz{x}{x}{\suc(\suc(f(x)))}}.
$$
It behaves as $\termdbl$ when fed with $0$, but it diverges when it receives a positive natural
number as an argument. But look: $\termdiv$ is not so different from $\termdbl$. Indeed, the second can
be obtained from the first by feeding not $x$ but $\pred(x)$ to~$f$.
And any type system in which $\termdbl$ and $\termdiv$ are somehow recognised as being 
fundamentally different must be able to detect the presence of $\pred$ in $\termdbl$ and deduce termination
from it.
Indeed, sized types \cite{BartheGR08} and dependent types \cite{Xi01} are able to do so.
Going further, we could ask the type system to be able not only to guarantee termination, 
but also to somehow evaluate the time or space consumption of programs.
For example, we could be interested in knowing that $\termdbl$ takes a polynomial number 
of steps to be evaluated on any natural number, and actually some type systems able to control 
the complexity of higher-order programs exist. Good examples are type systems for amortised 
analysis \cite{HOAA10,HAH11} or those using ideas from linear logic \cite{BaillotT09,BaillotGaboardiMogbil09esop}:
in all of them, linearity plays a key role.

\dlpcfn~\cite{DLG11} combines some of the ideas presented above with the principles of bounded linear logic (\BLL~\cite{GirardSS92}):
the cost of evaluating a term is measured by counting how many times function arguments 
need to be copied during evaluation, and different copies can be given distinct, although uniform, types.
Making this information explicit in types permits to compute the cost step by step during
the type derivation process. Roughly, typing judgements in \dlpcfn\ are statements like
$$
\vdash_{\J(a)}\t\ :\ !_{\inb}\ \mathsf{Nat}[a]\multimap\mathsf{Nat}[\I(a)],
$$
where $\I$ and $\J$ depend on $a$ and $n$ is a natural number capturing the number of
times $\t$ uses its argument. But this is not sufficient:
analogously to what happens in \BLL, \dlpcfn\ makes types more parametric.
A type like $!_{\inb}\ `s\multimap `t$ is replaced by the 
more parametric type $!_{a<\inb}`s\multimap`t$, which tells us
that the argument will be used~$n$ times, and each instance has type~$`s$ 
\emph{where, however} the variable~$a$ is instantiated with a value less than~$n$.
This allows to type each copy of the argument differently but uniformly, since all 
instances of~$`s$ have the same \PCF\ skeleton.
This form of \emph{uniform linear dependence} is actually crucial in obtaining
the result which makes \dlpcfn\ different from similar type systems, namely completeness.
As an example, $\termdbl$ can be typed as follows in \dlpcfn:
$$
\vdash_{a}^{\ep}\termdbl:
\btyp[b]{a+1}{\mathsf{Nat}[a]}\multimap\mathsf{Nat}[2\times a].
$$
This tells us that the argument will be used $a$ times by $\termdbl$, namely a number of times
equal to its value. And that the cost of evaluation will be itself proportional to $a$. 

\subsection{Why Another Type System?}
The theory of $`l$-calculus is full of interesting results, one of them being the so-called Church-Rösser property:
both $`b$ and $`b`h$ reduction are confluent, \ie if you fire two distinct redexes 
in a $`l$-term, you can always ``close the diagram'' by performing one \emph{or more} rewriting steps.
This, however, is not a \emph{local} confluence result, and as such does \emph{not} imply that all reduction strategies are 
computationally equivalent. Indeed, some of them are normalising (like normal-order evaluation) while some others are 
not (like innermost reduction). But how about efficiency? 

On the one hand, it is well known that optimal reduction \emph{is} indeed possible~\cite{Lamping90}, even if
it gives rise to high overheads~\cite{AspertiMairson98}. On the other, call-by-name can be highly inefficient.
Consider, as an example, the composition of $\termdbl$ with itself:
$$
\termdbltwo=`lx.\termdbl(\termdbl\ x).
$$
This takes quadratic time to be evaluated in the \KAM: the evaluation of $(\termdbl\ \nb)$ is repeated
a linear number of times, whenever it reaches the head position. This actually \emph{can} be seen from within \dlpcfn,
since 
$$
\vdash_{\J}^{\ep}\termdbltwo:!_{b<\I}\mathsf{Nat}[a]\multimap\mathsf{Nat}[4\times a].
$$
where both $\I$ and $\J$ are quadratic in $a$. Call-by-value solves this problem, at the price of not being normalising. 
Indeed, eager evaluation of $\termdbltwo$ when fed with a natural number~$n$ takes linear time in~$n$.
The relative efficiency of call-by-value evaluation, compared to call-by-name, is not a novelty: many modern functional
programming languages (like \textsf{OCaml} and \textsf{Scheme}) are based on it, while very few of them evaluate terms in call-by-name
order. 

For the reasons above, we strongly believe that designing a type system in the style of \dlpcfn, but able to deal with eager
evaluation, is a step forward applying linear dependent types to actual programming languages. 

  \subsection{Call-by-Value, Call-by-Name and Linear Logic}
\label{sec:cbvcbn}
Various notions of evaluation for the $\lambda$-calculus can be seen as translations of intuitionistic
logic (or of simply-typed $\lambda$-calculi) into Girard's linear logic. This correspondence has been
investigated in the specific cases of call-by-name (\cbn) and call-by-value (\cbv) reduction (\eg see the work
of Maraist et al.~\cite{MaraistOTW95}). In this section, we briefly introduce the main ideas behind 
the correspondence, explaining why linear logic has guided the design of \dlpcfv.

The general principle in such translations, is to guarantee that whenever a term \emph{can} possibly be
duplicated, it must be mapped to a box in the underlying linear logic proof.
In the \cbn\ translation (also called Girard's translation), \emph{any} argument to functions
can possibly be substituted for a variable and copied, so arguments are banged during the translation:
\begin{displaymath}
  (A"=>"B)^*~=~(!A^*)\multimap B^*
\end{displaymath}
Adding the quantitative bound on banged types (as explained in the previous section) gives rise to 
the type $\freccia{(\btyp{\I}{`s})}{`t}$ for functions (written $\freccia{[a<\I]`.`s}{`t}$ in~\cite{DLG11}).
In the same way, \emph{contexts} are banged in the \cbn\ translation: a typing judgement in \dlpcfn\ have
the following form:
\begin{displaymath}
  x_1:\,\btyp[a_1]{\I_1}{`s_1},\dots,x_n:\,\btyp[a_n]{\I_n}{`s_n}
  ~\vdash_{\J}~
  \t:`t.
\end{displaymath}
In the \cbv\ translation, $`b$-reduction should be performed only if the argument is a value.
Thus, arguments are not automatically banged during the translation but values are, so that the $`b$-reduction 
remains blocked until the argument reduces to a value. In the $\lambda$-calculus values are functions, hence 
the translation of the intuitionistic arrow becomes
\begin{displaymath}
  (A"=>"B)^{`o}~=~!(A^{`o}\multimap B^{`o}).
\end{displaymath}
Function types in \dlpcfv\ then become $\btyp{\I}{(\freccia{`s}{`t})}$, and a judgement has the form
\begin{math}
  x_1:`s_1,\dots,x_n:`s_n~\vdash_{\J}~\t:`t.
\end{math}
The syntax of types varies fairly much between \dlpcfn\ to \dlpcfv, and consequently
the two type systems are different, although both of them are greatly inspired by linear logic.

In both cases, however, the ``target'' of the translation is not the whole of \ILL, but rather a restricted
version of it, namely \BLL, in which the complexity of normalisation is kept under control by shifting from
unbounded, infinitary, exponentials to finitary ones.
For example, the \BLL\ \textit{contraction} rule allows to merge the first~\I\ copies of~$\A$, and the following~\J\ ones into 
the first~$\I+\J$ copies of~$A$:
\begin{displaymath}
  \frac{`G,\btyp{\I}{\A},\btyp{\J}{\A\isubst{\I+a}}\vdash\B}
    {`G,x:\,\btyp{\I+\J}{\A}\vdash\B}
\end{displaymath}
We write $`s\uplus`t=\btyp{\I+\J}{\A}$ if $`s=\btyp{\I}{\A}$ and \mbox{$`t=\btyp{\J}{\A\isubst{\I+a}}$}. 
Any time a contraction rule is involved in the \cbv\ translation of a type derivation, a sum~$\uplus$ appears at the same place in the corresponding \dlpcfv\ derivation.
Similarly, the \textit{dereliction} rule allows to see any type as the
first copy of itself:
\begin{displaymath}
  \frac{`G,\A\isubst{\zero}\vdash\B}{
  `G,\btyp{\one}{\A}\vdash\B}
\end{displaymath}
hence any dereliction rule appearing in the translation of a typing judgement tells us that the corresponding type is copied once.
Both contraction and dereliction appear while typing an application in \dlpcfv: the \PCF\ typing rule
\begin{displaymath}
    \frac{`G~\vdash~\t:~\A⇒\B\quad`G~\vdash~\u:~\A
    }{`G~\vdash~\t\u:~\B}
\end{displaymath}
corresponds to the following \ILL\ proof:
$$
  \def\defaultHypSeparation{\hskip 5pt}
  \def\ScoreOverhang{4pt}
  \AxiomC{$\scriptstyle
    z:~\trad{\A}\multimap\trad{\B}\vdash     z:~\trad{\A}\multimap\trad{\B}$}
  \RightLabel{$ ^{\mathit{der}}$}
  \UnaryInfC{$\scriptstyle
    !z:~!(\trad{\A}\multimap\trad{\B})\vdash  z:~\trad{\A}\multimap\trad{\B}$}
  \AxiomC{$\scriptstyle
    \trad{`G}\vdash \trad{\t}:~!(\trad{\A}\multimap\trad{\B})$}
  \BinaryInfC{$\scriptstyle
    \trad{`G}\vdash \trad{\t}:~\trad{\A}\multimap\trad{\B}$}
  \AxiomC{$\scriptstyle
    \trad{`G}\vdash \trad{\u}:~\trad{\A}$}
  \BinaryInfC{$\scriptstyle
    \trad{`G},~\trad{`G}\vdash \trad{\t}\trad{\u}:~\trad{\B}$}
  \RightLabel{$ ^\mathit{contr}$}
  \LeftLabel{~$\scriptstyle ^{(\trad{`G}=!`G')}$}
  \UnaryInfC{$\scriptstyle
    \trad{`G}\vdash \trad{\t}\trad{\u}:~\trad{\B}$}
  \DisplayProof
$$
which becomes the following, when appropriately decorated according to the principles
of~\BLL\ (writing~$\A_0$ and~$\B_0$ for $\A\isubst{\zero}$ and $\B\isubst{\zero}$):
$$
  \def\defaultHypSeparation{\hskip 1pt}
  \def\ScoreOverhang{1pt}
  \AxiomC{$\scriptstyle
    z:~\trad{\A_0}\multimap\trad{\B_0}\vdash     z:~\trad{\A_0}\multimap\trad{\B_0}$}
  \RightLabel{$ ^{\mathit{der}}$}
  \UnaryInfC{$\scriptstyle
    !z:~\btyp{\one}{(\trad{\A}\multimap\trad{\B})}\vdash  z:~\trad{\A_0}\multimap\trad{\B_0}$}
  \AxiomC{$\scriptstyle
    \trad{`G}\vdash \trad{\t}:~\btyp{\one}{(\trad{\A}\multimap\trad{\B})}$}
  \BinaryInfC{$\scriptstyle
    \trad{`G}\vdash \trad{\t}:~\trad{\A_0}\multimap\trad{\B_0}$}
  \AxiomC{$\scriptstyle
    \trad{`G}\vdash \trad{\u}:~\trad{\A_0}$}
  \BinaryInfC{$\scriptstyle
    \trad{`G},~\trad{`G}\vdash \trad{\t}\trad{\u}:~\trad{\B_0}$}
  \RightLabel{$ ^\mathit{contr}$}
  \LeftLabel{~$\scriptstyle ^{(\trad{`G}=!`G')}$}
  \UnaryInfC{$\scriptstyle
    \trad{`G}\uplus\trad{`G}\vdash \trad{\t}\trad{\u}:~\trad{\B_0}$}
  \DisplayProof
$$
This \cbv\ translation of the application rule hence leads to the typing rule for applications in \dlpcfv:
\begin{displaymath}
  \frac{`G\vdash_{\K}\t:\btyp{\one}{(\freccia{`s}{`t})}\quad
  `D\vdash_{\H}\u:`s\isubst{\zero}}{
  `G\uplus`D\vdash_{\K+\H}\t\u:`t\isubst{\zero}}
\end{displaymath}
The same kind of analysis enables to derive the typing rule for abstractions (whose call-by-value 
translation requires the use of a promotion rule) in \dlpcfv:
\begin{displaymath}
  \frac{`G,x:`s\vdash_{\K}\t:`t}{\sum_{a<\I}`G\vdash_{\I+\sum_{a<\I}\K}
    `lx.\t:\btyp{\I}{(\freccia{`s}{`t})}}
\end{displaymath}
One may wonder what~\I\ represents in this typing rule, and more generally in a judgement such as
\begin{displaymath}
  `G\vdash_{\K}\t:~\btyp{\I}{\A}.
\end{displaymath}
This is actually the main new idea of \dlpcfv:
such a judgement intuitively means that the value to which~\t\ reduces will be used~\I\ times by the environment.
If~\t\ is applied to an argument~\u, then~\t\ must reduce to an abstraction~$`lx.\s$, that is destructed 
by the argument without being duplicated. In that case, $\I=\one$, as indicated by the application typing rule.
On the opposite, if~\t\ is applied to a function~$`lx.\u$, then the type of this function must be of the form
(up to a substitution of~$b$)
$\btyp[b]{\one}{(\freccia{\btyp{\I}{\A}}{`t})}$. This means that~$`lx.\u$ 
uses~\I\ times its arguments, or, that~$x$ can appear at most~\I\ times in the reducts of~\u.

This suggests that the type derivation of a term is not unique in general:
whether a term~\t\ has type~$\btyp{\I}{\A}$ or~$\btyp{\J}{\A}$ depends on the use we want to make of~\t.
This intuition will direct us in establishing the typing rules for the other \PCF\ constructs
(namely conditional branching and fixpoints).


%% file: formalism.tex
\section{\dlpcfv, Formally}
\label{sec:formal}
In this section, the language of programs and a type system \dlpcfv\ for it will be introduced formally.
While programs are just terms of a fairly standard $\lambda$-calculus (which is very similar to Plotkin's
\PCF), types may include so-called \emph{index terms}, which are first-order
terms denoting natural numbers by which one can express properties about the extensional and intentional behaviour
of programs. 

\subsection{Index Terms and Equational Programs}
\label{sec:index}

Syntactically, index terms are built either from function symbols from a given untyped signature $\Theta$ or
by applying any of two special term constructs:
\begin{displaymath}
  \I,\J,\K \quad::=\quad
  a~~|~~\f(\I_1,\ldots,\I_{n}) 
  ~~|~~\displaystyle{\sum_{a< \I}\J}~~|~~ \displaystyle{\Dfscomb{a}{\I}{\J}{\K}}.
\end{displaymath}
Here, $\f$ is a symbol of arity $n$ from $\Theta$ and $a$ is a variable drawn from a set $\vars$ of \emph{index variables}.
We assume the symbols $0$, $1$ (with arity $0$) and $+$, $\mnu$ (with arity $2$) are always part of $\Theta$.
An index term in the form $\sum_{a< \I}\J$ is a \emph{bounded sum}, while one in the form $\Dfscomb{a}{\I}{\J}{\K}$
is a \emph{forest cardinality}.
For every natural number $n$, the index term $\inb$ is just $\underbrace{1+1+\ldots+1}_{\mbox{$n$ times}}$.

Index terms are meant to denote natural numbers, possibly depending on the (unknown) values of variables.
Variables can be instantiated with other index terms, \eg $\I\isubst{\J}$. 
So, index terms can also act as first order functions. 
What is the meaning of the function symbols from $\Theta$? It is the one induced by  an equational program $\ep$.
Formally, an \emph{equational program} $\ep$ over a signature $\Theta$ 
is a set of equations in the form $\I=\J$ where both $\I$ and $\J$
are index terms. We are interested in equational programs guaranteeing that, whenever
symbols in $\Theta$ are interpreted as partial functions
over $\int$ and $0$, $1$, $+$ and $\mnu$ are interpreted in the
usual way, the semantics of any function symbol 
$f$ can be uniquely determined from $\ep$.
This can be guaranteed by, for example, taking $\ep$ as
an Herbrand-G\"odel scheme~\cite{Odifreddi} or as an orthogonal constructor
term rewriting system~\cite{BaaderNipkow}. 
The definition of index terms is parametric on $\Theta$ and $\ep$:
this way one can tune our type system from a highly undecidable but truly 
powerful machinery down to a tractable but less expressive formal system. 

What about the meaning of bounded sums and forest cardinalities? The first is very intuitive:
the value of $\sum_{a< \I}\J$ is simply the sum of all
possible values of $\J$ with $a$ taking the values from $0$ up to
$\I$, excluded. Forest cardinalities, on the other hand, require some effort to be described.
Informally, $\Dfscomb{a}{\I}{\J}{\K}$ is an
index term denoting the number of nodes in a forest composed of $\J$ trees described using $\K$.
All the nodes in the forest are (uniquely) identified by natural numbers.
These are obtained by consecutively visiting each tree in pre-order, starting from $\I$.
The term~$\K$ has the role of describing the number of children of each forest node, 
\eg the number of children of the node $\zero$ is $\K\isubst{\zero}$.
More formally, the meaning of a forest cardinality is defined by the following two equations:
\begin{align*}
  \Dfscomb{a}{\I}{0}{\K}&=0 \\
  \Dfscomb{a}{\I}{\J+1}{\K}&=
  \left(\Dfscomb{a}{\I}{\J}{\K}\right)+1+
  \left(\Dfscomb{a}{\I+1+\Dfscomb{a}{\I}{\J}{\K}}
    {\K\isubst{\I+\Dfscomb{a}{\I}{\J}{\K}}}{\K}\right)
\end{align*}
The first equation says that a forest of $0$ trees contains no nodes.
The second one tells us that a forest of $\J+1$ trees contains:
\begin{itemize} 
\item
  The nodes in the first $\J$ trees;
\item
  plus the nodes in the last tree, which are just one plus the nodes   in the immediate 
  subtrees of the root, considered themselves as a forest.
\end{itemize}
To better understand forest cardinalities, consider the following forest comprising two trees:
$$
\scalebox{0.8}{
  \xymatrix{
    & &  0          &   \\
    & & 1\ar@{-}[u] &   \\
    \empty&2\ar@{-}[ur] &5\ar@{-}[u] & 6\ar@{-}[ul]\\
    3\ar@{-}[ur]& &  4\ar@{-}[ul]&7\ar@{-}[u]
    \save "4,1"."3,4"*[F.]\frm{--}\restore 
  }

  \xymatrix{
    & & 8 &   \\
    & 9\ar@{-}[ur] & & 11\ar@{-}[ul]\\
    & 10\ar@{-}[u] & & 12\ar@{-}[u] \\
  }
}
$$
It is well described by an index term $\K$ with a free index variable $a$ such that 
$\K\isubst{1}=3$;
$\K\isubst{n}=2$ for $n\in\{2,8\}$;
$\K\isubst{n}=1$ when
$n\in\{0,6,9,11\}$;
and
$\K\isubst{n}=0$ when
$n\in\{3,4,5,7,10,12\}$. That is, $\K$ describes the number of children of each node. Then 
$\Dfscomb{a}{\zero}{2}{\K}=13$ since it takes into account the entire forest;
$\Dfscomb{a}{\zero}{1}{\K}=8$ since it takes into account only the leftmost tree;
$\Dfscomb{a}{8}{1}{\K}=5$ since it takes into account only the second tree of the forest;
finally, 
$\Dfscomb{a}{2}{3}{\K}=6$ since it takes into account only the three trees (as a forest) 
within the dashed rectangle.

One may wonder what is the role of forest cardinalities in the type system. Actually, 
they play a crucial role in the treatment of recursion, where the unfolding of recursion
produces a tree-like structure whose size is just the number of times the (recursively
defined) function will be used \emph{globally}. 
Note that the value of a forest cardinality could also be undefined. 
For instance, this happens when infinite trees, corresponding to 
diverging recursive computations, are considered.

The expression $\itp{\I}$ denotes the meaning of $\I$,
defined by induction along the lines of the previous discussion, 
where $`r:\vars\to\int$ is an assignment and~$\ep$ is an equational program giving meaning to the function symbols in~$\I$.
Since~$\ep$ does not necessarily interpret such symbols as \emph{total} functions, and moreover, the value of a forest cardinality can be undefined, $\itp{\I}$ can be undefined itself.
A \emph{constraint} is an inequality in the form $\I\leq\J$.
Such a constraint is \emph{true} \condinc{(or \emph{satisfied})}{} in an assignment~$`r$
if~$\itp{\I}$ and~$\itp{\J}$ are \emph{both} defined and the first is smaller or equal to the latter.
Now, for a subset~$\fiv$ of~$\vars$, and for a set~$\ictx$ of constraints involving variables in~$\fiv$, the expression
\begin{displaymath}
  \ijudg{\fiv}{\ictx}{\I\leq\J}
\end{displaymath}
denotes the fact that the truth of $\I\leq\J$ \emph{semantically} follows from the truth of the constraints in $\ictx$. 
To denote that~\I\ is well defined for~\ep\ and any valuation~$`r$ satisfying~\ictx, we may write \ijudg{\fiv}{\ictx}{\I`|} instead of \ijudg{\fiv}{\ictx}{\I\leq\I}.

\subsection{Programs}
\label{sec:prg}
\emph{Values} and \emph{terms} are generated by the following grammar:
\begin{displaymath}
  \begin{array}{c@{\quad}rl}
    \text{Values:} & \v, \w ::= & \nb~|~`lx.\t ~|~\fix{t} \\
    \text{Terms:} & \s,\t,\u ::= & 
    x~|~\v~|~\t\u~|~\suc(\t)~|~\pred(\t) \\
    &&|\ifz{\t}{\u}{\s}  
  \end{array}
\end{displaymath}
Terms can be typed with a well-known type system called \PCF:
types are those generated by the basic type $\NatPCF$ and the binary type constructor $\arr$.
Typing rules are standard (see~\cite{EV}). 
A notion of (weak) call-by-value reduction $\tov$ can be easily defined: take the
rewriting rules in Figure~\ref{fig:cbv-red} and close them under all applicative contexts.
\begin{figure}                                             %
  \begin{center}
    \condinc{\fbox{\begin{minipage}{0.75\textwidth}
          $$
          \begin{array}{r@{~~\tov~~}l}
            (`lx.\t)~\v & \t\tsubst{\v}\\
            \suc(\nb)& \nb[n+1]\\
            \pred(\nb[n+1])& \nb\\
            \pred(\nb[0])& \nb[0]\\
            \ifz{\nb[0]}{\t}{\u} & \t\\
            \ifz{\nb[n+1]}{\t}{\u} & \u\\
            (\fix{\t})~\v & (\t\tsubst{\fix{\t}})~\v\\      
          \end{array}
          $$
        \end{minipage}}
    }{
      $$
      \begin{array}{r@{~~\tov~~}l}
        (`lx.\t)~\v & \t\tsubst{\v}\\
        \suc(\nb)& \nb[n+1]\\
        \pred(\nb[n+1])& \nb\\
        \pred(\nb[0])& \nb[0]\\
        \ifz{\nb[0]}{\t}{\u} & \t\\
        \ifz{\nb[n+1]}{\t}{\u} & \u\\
        (\fix{\t})~\v & (\t\tsubst{\fix{\t}})~\v\\      
      \end{array}
      $$
    }
  \end{center}
  \caption{Call-by-value reduction of \PCF\ terms.}
  \label{fig:cbv-red}
\end{figure}                                               %
A term $\t$ is said to be a \emph{program} if it can be given the \PCF\ type 
$\NatPCF$ in the empty context. 
The \emph{multiplicative size} $\ms{\t}$ of a term $\t$ is defined as follows:
\begin{displaymath}
  \begin{array}[t]{r@{~=~}l}
  \ms{\nb}=\ms{`lx.\t}=\ms{\fix{t}}&0~;\\
  \ms{x}&2~;\\
  \ms{\t\u}&\ms{\t}+\ms{\u}+2~;\\
  \ms{\suc(\t)}&\ms{\t}+2\\
  \ms{\pred(\t)}&\ms{\t}+2~;\\
  \ms{\ifz{\t}{\u}{\s}}&\ms{\t}+\ms{\u}+\ms{\s}+2.
\end{array}
\end{displaymath}

Notice that the multiplicative size of a term~\t\ is less or equal than its
size $\ts{\t}$ (which is defined inductively, similarly to $\ms{\t}$, except 
for values: $\ts{\nb}=2$, and $\ts{\fix{\t}}=\ts{`lx.\t}=\ts{\t}+2$).
Values are not taken into account by the multiplicative size.
Indeed, the evaluation of terms (\cf\ Section~\ref{sec:cek}) consists first in \emph{scanning} a term until a value is reached (and the cost of this step is measured by the multiplicative size).
Then this value is either destructed (\eg\ when a lambda abstraction is given an argument), either duplicated (\eg\ when it is itself an argument of a lambda abstraction).
The cost of this second step will be measured by the type system~\dlpcfv.

\subsection{The Type System}
\label{sec:typ}

\paragraph{The Language of Types}
\begin{figure*}                                            %
  \centering
  \condinc{\fbox{
      \begin{minipage}{.97\textwidth}
        \input{fig-subtype}
      \end{minipage}}
  }{
    \input{fig-subtype}
    \vspace{5pt}
  }
  \caption{Subtyping derivation rules of \dlpcfv.}
  \label{fig:tsub}
\end{figure*}                                              %
%
\begin{figure*}                                            %
  \condinc{
    \fbox{
      \begin{minipage}{.97\textwidth}
        \input{fig-type}
      \end{minipage}}
  }{
    \input{fig-type}
  }
  \caption{Typing rules of \dlpcfv.}
  \label{fig:typ}
\end{figure*}                                              %
The type system \dlpcfv\ can be seen as a refinement of \PCF\ obtained by a linear decoration of its type derivations. 
\emph{Linear} and \emph{modal types} are defined as follows:
\begin{align*}
\A, \B &::= \quad\freccia{`s}{`t} &\mbox{linear types}\\
`s,`t &::= \quad\mtyp{\I}{\A}~~|~~\Nat{\I}{\J} &\mbox{modal types}
\end{align*}
where $\I,\J$ range over index terms and $a$ ranges over index variables.
Modal types need some comments. Natural numbers are freely duplicable, so
$\Nat{\I}{\J}$ is modal by definition. As a first approximation, $\mtyp{\I}{\A}$ can be thought of 
as a universal quantification of $\A$, and so $a$ is bound in the linear type $\A$. 
Moreover, the condition $a<\I$ says that~$`s$ consists of all the instances of the linear type $\A$ 
where the variable~$a$ is successively instantiated with the values from $\zero$ to 
$\I\mnu\one$, \ie\ $\A\isubst{\zero},\ldots,\A\isubst{\I-\one}$. For those readers who are familiar 
with linear logic, and in particular with \BLL, the modal type $\mtyp{\I}{\A}$ is a generalisation 
of the \BLL\ formula $\btyp{p}{\A}$ to arbitrary index terms. As such it can be thought of as 
representing the type $\A\isubst{\zero}\otimes\cdots\otimes\A\isubst{\I-\one}$.
$\NatPCF[\I]$ is syntactic sugar for $\Nat{\I}{\I}$.
In the typing rules we are going to define, modal types need to be manipulated in an algebraic way.
For this reason, two operations on modal types need to be introduced.
The first one is a binary operation~$\uplus$ on modal types.
Suppose that $`s=\mtyp{\I}{\A\isubst[c]{a}}$
and that $`t=\mtyp[b]{\J}{\A\isubst[c]{\I+b}}$.
\condinc{
In other words, $`s$ consists of the first $\I$ instances of $\A$, \ie\ 
$\A\isubst[c]{\zero},\ldots,\A\isubst[c]{\I-\one}$ while $`t$ consists of the next $\J$ instances of $\A$,
\ie\ $\A\isubst[c]{\I+\zero},\ldots,\A\isubst[c]{\I+\J-\one}$.}{} 
Their \emph{sum} $`s\uplus`t$ is naturally defined as a modal type consisting of the first $\I+\J$ instances of $\A$, 
\ie\ $\mtyp[c]{\I+\J}{\A}$. Furthermore, $\Nat{\I}{\J}\uplus\Nat{\I}{\J}$ is just $\Nat{\I}{\J}$.
An operation of bounded sum on modal types can be defined by generalising the idea above:
suppose that 
$$
`s=\mtyp[b]{\J}{\A\isubst[c]{b+\sum_{d<a}\J\isubst{d}}}.
$$
Then its \emph{bounded sum} $\sum_{a<\I}`s$ is just $\mtyp[c]{\sum_{a<\I}\J}{\A}$.
Finally, $\sum_{a<\I}\Nat{\J}{\K}=\Nat{\J}{\K}$, provided $a$ is not free in $\J$ nor in 
$\K$.
\paragraph{Subtyping}
Central to \dlpcfv\ is the notion of subtyping.
An inequality relation $`<$ between (linear or modal) types can be defined using the formal system in Figure~\ref{fig:tsub}. 
This relation corresponds to lifting index inequalities at the type level.
Please observe that $`<$ is a pre-order, \ie, a reflexive and transitive relation.
\paragraph{Typing}
A typing judgement is of the form
\begin{displaymath}
  \judg{\fiv}{\ictx}{`G}{\K}{\t}{`t}~,
\end{displaymath}
where~\K\ is the \emph{weight} of~\t, that is (informally) the maximal number of 
substitutions involved in the \cbv\ evaluation of~\t.
\ictx\ is a set of constraints (\cf Section~\ref{sec:index}) that we call the 
\emph{index context}, and~$`G$ is a context assigning a modal type to (at least) each free variable of~\t.
Both sums and bounded sums are naturally extended from modal types to contexts 
(with, for instance, $\{x:`s;y:`t\}\uplus\{x:`s',z:`t'\}=\{x:`s\uplus`s';y:`t;z:`t'\}$).
There might be free index variables in~$\ictx,`G,`t$ and~\K, all of them from~\fiv.
Typing judgements can be derived from the rules of Figure~\ref{fig:typ}.

Derivation rules for abstractions and applications have been informally presented in Section~\ref{sec:cbvcbn}.
The other ones are then intuitive, except the derivation rule for typing~$\fix{\t}$, that is worth an explanation:
to simplify, assume we want to type only one copy of its type (that is, $\K=\one$).
To compute the weight of~\fix{\t}, we need to know the number of times~\t\ 
will be copied during the evaluation, that is the number of nodes in 
the tree of its recursive calls. This tree is described by~\I\ (as explained 
in Section~\ref{sec:index}), since each occurrence of~$x$ in~\t\ stands for a recursive call.
It has, say, $\H=\fc{0}{\one}{\I}$ nodes. At each node~$b$ of this tree, 
the~$a^{\mathit{th}}$ occurrence of~$x$ will be replaced by the~$a^{\mathit{th}}$ 
son of~$b$, \ie\ by $b+1+\fc{b+1}{a}{\I}$.
The types have to match, and that is what the second premise expresses.
Finally, the type of~\fix{\t} is the type of the ``main'' copy of~\t, at the root of the tree (\ie, at $b=0$).
The weight counts all the recursive calls (\ie,~\H) plus the weight of each copy of~\t\ (\ie, the weight of~\t\ for each $b<\H$).

Last, the subsumption rule allows to relax the precision standards of a typing judgement.
One can also restrict the inequalities on indexes to equalities in this rule, and thereby construct only \textit{precise} typing judgements.
Observe that the set of all rules but this one is syntax directed.
Moreover the subsumption rule preserves the \PCF\ skeleton of the types, and so the type system is itself syntax directed \emph{up to} index inequalities.

\subsection{An Abstract Machine for \PCF}
\label{sec:cek}

The call-by-value evaluation of \PCF\ terms can be faithfully captured
by an abstract machine in the style of \CEK~\cite{FelleisenF87}, which will be introduced in this section.

\begin{figure*}                                            %
  \centering
  \condinc{
    \fbox{
      \begin{minipage}{.97\textwidth}
        \input{fig-cekv}
      \end{minipage}}
  }{
    \input{fig-cekv}
  }    
  \caption{\cek\ evaluation rules for value closures.}
  \label{fig:cek-val}
\end{figure*}                                              %
%
\begin{figure*}                                            %
  \centering
  \condinc{
    \fbox{
      \begin{minipage}{.97\textwidth}
        \input{fig-cek}
      \end{minipage}}
  }{
    \input{fig-cek}
  }
  \caption{\cek\ contextual evaluation rules.}
  \label{fig:cek-clo}
\end{figure*}                                              %
The internal state of the \cek\ machine consists of a closure and a stack,
interacting following a set of rules. Formally, a \emph{value closure}
is a pair $\clov=\clo{\v}$ where $\v$ is a value and $\env$ is
an \emph{environment}, itself a list of assignments of value closures to
variables:
$$
\env::=∅~|~\envelt{\clov}\cdot\env.
$$
A \emph{closure} is a pair $\cloc=\clo{\t}$ where $\t$ is
a term (and not necessarily a value). \emph{Stacks} are terms
from the following grammar:
\begin{align*}
\stack::=&\emptystack ~|~ \stackfun{\clov}
    ~|~ \stackarg{\cloc}\\
    & ~|~ \stackif{\t}{\u}{\env} ~|~ \stacks ~|~ \stackp.
\end{align*}
A \emph{process} $\process$ is a pair $\proc{\cloc}{\stack}$
of a closure and a stack.

Processes evolve according to a number of rules. Some of
them (see Figure~\ref{fig:cek-val}) describe how the \cek\ machine
evolves when the first component of the process is a value
closure.
Other rules (see Figure~\ref{fig:cek-clo}) prescribe the evolution of \cek\ in all the other cases.

The following tells us that \cek\ is an adequate methodology to evaluate
\PCF\ terms:
\begin{proposition}[Adequacy]\label{prop:eq}
If $\t$ is a \PCF\ term of type $\NatPCF$, then 
$\t\tov^*\nb$ iff $(\proc{\clo[\emptyset]{\t}}{\emptystack})\tocek^*(\proc{\clo[\emptyset]{\nb}}{\emptystack})$.
\end{proposition}

\paragraph{Weights and \cek\ Machine}
As it will be formalised in Section~\ref{sec:sound}, an upper bound for the evaluation of a given term 
in the \cek\ machine can be obtained by multiplying its weight and its size. This results can be explained as follows:
we have seen (in Section~\ref{sec:typ}) that its weight represents the maximal number of substitutions in 
its \cbv\ evaluation, and thereby the maximal number of steps of the form
\condinc{
\begin{align}
  \label{eq:fun-rul1}
  \proc{\clov}{\stackfun{\clo{`lx.\t}}}&\tocek   \proc{\clo[\envelt{\clov}\cdot\env]{\t}}{\stack}\\
  \label{eq:fun-rul2}
  \proc{\clov}{\stackfun{\clo{\fix{\t}}}}&\tocek
    \proc{\clo[\envelt{\clo{\fix{\t}}}\cdot\env]{\t}}{\stackarg{\clov}}
\end{align}
}{
\begin{align*}
  \proc{\clov}{\stackfun{\clo{&`lx.\t}}}\tocek \proc{\clo[\envelt{\clov}\cdot\env]{\t}}{\stack}\\
  \proc{\clov}{\stackfun{\clo{&\fix{\t}}}} \tocek\\
    &\proc{\clo[\envelt{\clo{\fix{\t}}}\cdot\env]{\t}}{\stackarg{\clov}}
\end{align*}
}
in its evaluation with the \cek.
Between two such steps, the use of the other rules is not taken into account by the weight;
however these other rules make the \emph{size} of the process to decrease.



%% file: fig-subtype.tex
      \begin{math}
        \infer[]
              {\sjudg{\fiv}{\ictx}{\Nat{\I}{\J}`<\Nat{\K}{\H}}}
              {
                \begin{array}{c}
                  \ijudg{\fiv}{\ictx}{\K\leq\I} \\
                  \ijudg{\fiv}{\ictx}{\J\leq\H} \\
                \end{array}
              }
              \condinc{\;}{\qquad}
              \infer[]
                    {\sjudg{\fiv}{\ictx}{\freccia{`s}{`t}`<\freccia{`s'}{`t'}}}
                    {
                      \begin{array}{c}
                        \sjudg{\fiv}{\ictx}{`s'`<`s}\\
                        \sjudg{\fiv}{\ictx}{`t`<`t'}
                      \end{array}
                    }
              \condinc{\;}{\qquad}
              \infer[]
                    {\sjudg{\fiv}{\ictx}{\mtyp{\I}{\A} `<
\mtyp{\J}{\B}}}
                    {
                      \begin{array}{r@{\,}l}
                        \sjudg{(a,\fiv)}{(a<\J,\ictx)&}{\A`<\B}\\
                        \ijudg{\fiv}{\ictx&}{\J\leq\I}
                      \end{array}
                    }
      \end{math}


%% file: fig-type.tex
\begin{displaymath} 
  \frac{}{\judg{\fiv}{\ictx}{`G,x:`s}{\zero}{x}{`s}}(\textit{Ax})
  \qquad
  \frac{\judg{\fiv}{\ictx}{`G}{\I}{t}{`s}
    \qquad
    \sjudg{\fiv}{\ictx}{`D`<`G}
    \qquad
    \sjudg{\fiv}{\ictx}{`s`<`t}
    \qquad
    \ijudg{\fiv}{\ictx}{\I\leq\J}
  }{\judg{\fiv}{\ictx}{`D}{\J}{t}{`t}}(\textit{Subs})
\end{displaymath}
\begin{displaymath}
  \frac{
    \judg{(a,\fiv)}{(a<\I,\ictx)}{`G,x:`s}{\K}{\t}{`t}
  }{
    \judg{\fiv}{\ictx}{\sum_{a<\I}`G}{\I+\sum_{a<\I}\K}{`lx.\t}{\mtyp{\I}{\freccia{`s}{`t}}}
  }
  (\multimap)
  \condinc{\quad}{\qquad}
  \frac{
    \judg{\fiv}{\ictx}{`G}{\K}{\t}{\mtyp{\one}{\freccia{`s}{`t}}}
    \qquad 
    \judg{\fiv}{\ictx}{`D}{\H}{\u}{`s\isubst{\zero}}
  }{
    \judg{\fiv}{\ictx}{`G\uplus`D}{\K+\H}{\t\u}{`t\isubst{\zero}}
  }(\textit{App})
\end{displaymath}   
\begin{displaymath}
  \quad
  \frac{
    \judg{\fiv}{\ictx}{`G}{\M}{\t}{\Nat{\J}{\K}}
    \qquad 
    \judg{\fiv}{(\J\leq\zero,\ictx)}{`D}{\N}{\u}{`t}
    \qquad
    \judg{\fiv}{(\K\geq\one,\ictx)}{`D}{\N}{\s}{`t}
  }{
    \judg{\fiv}{\ictx}{`G\uplus`D}{\M+\N}{\ifz{\t}{\u}{\s}}{`t}
  }(\textit{If})
\end{displaymath}
\begin{displaymath}
  \frac{}{\judg{\fiv}{\ictx}{`G}{\zero}{\nb}{\Nat{\inb}{\inb}}}(n)
  \qquad
  \frac{
    \judg{\fiv}{\ictx}{`G}{\M}{\t}{\Nat{\I}{\J}}
  }{
    \judg{\fiv}{\ictx}{`G}{\M}{\suc(\t)}{\Nat{\I+\one}{\J+\one}}
  }(s)
  \qquad
  \frac{
    \judg{\fiv}{\ictx}{`G}{\M}{\t}{\Nat{\I}{\J}}
  }{
    \judg{\fiv}{\ictx}{`G}{\M}{\pred(\t)}{\Nat{\I\mnu\one}{\J\mnu\one}}
  }(p)
\end{displaymath}
\condinc{
  \begin{displaymath}
    \frac{
      \begin{array}{r@{\,}l}
        \judg{(b,\fiv)}{(b<\H,\ictx)}{
          `G,x:\mtyp{\I}{\A}&}{\J}{\t}{\mtyp{1}{\B}}
        \\
        \sjudg{(a,b,\fiv)}{(a<\I,b<\H,\ictx)&}{
          \B\isubst{\zero}\isubst[b]{\fc{b+1}{a}{\I}+b+1}`<\A}
      \end{array}
    }{\judg{\fiv}{\ictx}{\sum_{b<\H}`G}{
        \H+\sum_{b<\H}\J}{\fix{\t}}{\mtyp{\K}{
          \B\isubst{\zero}\isubst[b]{\fc{0}{a}{\I}}}}}
    (\textit{Fix})
  \end{displaymath}
}{
  \begin{displaymath}
    \frac{
      \judg{(b,\fiv)}{(b<\H,\ictx)}{
        `G,x:\mtyp{\I}{\A}}{\J}{\t}{\mtyp{1}{\B}}
      \quad
      \sjudg{(a,b,\fiv)}{(a<\I,b<\H,\ictx)}{
        \B\isubst{\zero}\isubst[b]{\fc{b+1}{a}{\I}+b+1}`<\A}
    }{\judg{\fiv}{\ictx}{\sum_{b<\H}`G}{
        \H+\sum_{b<\H}\J}{\fix{\t}}{\mtyp{\K}{
          \B\isubst{\zero}\isubst[b]{\fc{0}{a}{\I}}}}}
    (\textit{Fix})
  \end{displaymath}
  \vspace{-10pt}

  \hfill{\small (where $\H=\fc{0}{\K}{\I}$)}
}


%% file: fig-cekv.tex
\begin{math}
  \begin{array}{c@{\quad\star\quad}c@{\qquad\tocek\qquad}c@{\quad\star\quad}c}
    \clov & \stackarg{\cloc} & \cloc & \stackfun{\clov} \\
    \clov & \stackfun{\clo{`lx.\t}} &
    \clo[\envelt{\clov}\cdot\env]{\t} & \stack \\
    \clov & \stackfun{\clo{\fix{\t}}} & \clo[\envelt{\clo{\fix{\t}}}\cdot\env]{\t} & \stackarg{\clov} \\
    \clo[\env']{\nb[0]} & \stackif{\t}{\u}{\env} & \clo{\t} &\stack\\
    \clo[\env']{\nb[n\!+\!1]} & \stackif{\t}{\u}{\env} & \clo{\u} & \stack \\
    \clo{\nb} & \stacks & \clo[\emptyset]{\nb[n\!+\!1]} & \stack \\
    \clo{\nb} & \stackp & \clo[\emptyset]{\nb[n\!-\!1]} & \stack \\
  \end{array}
\end{math}


%% file: fig-cek.tex
\begin{math}
  \begin{array}{c@{\quad\star\qquad}c@{\qquad\tocek\quad\qquad}c@{\quad\qquad\star\qquad}c}
    \clo{x} & \stack & \env(x) & \stack \\
    \clo{\t\u} & \stack & \clo{\t} & \stackarg{\clo{\u}} \\
    \clo{\suc(t)} & \stack & \clo{\t} & \stacks \\
    \clo{\pred(t)} & \stack & \clo{\t} & \stackp \\
    \clo{\ifz{\t}{\u}{\s}} & \stack & \clo{\t} & \stackif{\u}{\s}{\env} \\
  \end{array}
\end{math}


%% file: examples.tex
\section{Examples}\label{sec:examples}
In this section we will see how to type some ``real life'' functions in \dlpcfv, and what is the cost associated to them.

\paragraph{Addition}
In \PCF, addition can be computed as follows:
\begin{displaymath}
  \termadd~=~\fix[f]{`lyz.\ifz{y}{z}{\suc(f\,\pred(y)\,z)}}~,
\end{displaymath}
and has \PCF\ type
\begin{math}
  \NatPCF⇒\NatPCF⇒\NatPCF.
\end{math}
A brief analysis of its evaluation, if we apply it to two values~\v\
and~\w\ in~\NatPCF, indicates that a correct annotation for this type
in \dlpcfv\ would be
\begin{displaymath}
  \mtyp{\one}{
    \left(\freccia{\NatU{\f}}{
        \mtyp[c]{\one}{\left(
            \freccia{\NatU{\g}}{\NatU{\f+\g}}\right)}}
    \right)}
\end{displaymath}
where~\f\ and~\g\ are constant symbols representing the values of~\t\ and~\u\ respectively.
Since we directly apply~\termadd, without copying this function, the index variables~$a$ and~$c$ are bounded with~\one.
This type is indeed derivable for~\termadd\ in \dlpcfv, assuming that the equational program~\ep\ is powerful enough to assign the following meaning to the corresponding index (they all depend on a free index variable~$b$):
\begin{align*}
\I~&=~\mbox{if $b<\f$ then $\one$ else $\zero$};\\
\J~&=~\f-b-1;\\
\H~&=~\f-b;\\
\K~&=~\f-b+1.
\end{align*}
The derivation is given in Figure~\ref{fig:typ-add}.
We omit all the subsumption steps, but the index equalities they use are easy to check given that the number of nodes in the tree of recursive calls is $\fc[b]{\zero}{\one}{\I}=\f+\one$.
The final weight is equal to~$3`*(\f+\one)$.
\begin{figure*}                                            %
  \centering
  \condinc{\footnotesize
    \fbox{
      \begin{minipage}{.97\linewidth}
        \input{fig-exadd}
      \end{minipage}}
  }{
    \input{fig-exadd}
    \vspace{10pt}
  }
  \caption{Typing derivation of~\termadd }
  \label{fig:typ-add}
\end{figure*}                                              %

\paragraph{Multiplication}
The multiplication can be easily defined using the addition:
\begin{displaymath}
  \termmult~=~\fix{`lyz.\ifz{y}{\nb[0]}{\termadd~z~(x\,\pred(y)\,z)}}.
\end{displaymath}
Taking the indexes~\I,\J,\H\ and~\K\ defined as in the previous paragraph, 
and using the typing judgement for~\termadd\ with~\f\ replaced by~\g\ and \g\ 
replaced by~$\J`*\g$, we can assign to~\termmult\ the type
\begin{displaymath}
  \mtyp{\one}{\freccia{\NatU{\f}}{
      \mtyp[c]{\one}{(\freccia{\NatU{\g}}{
          \NatU{\f`*\g}})}}}
\end{displaymath}
(see Figure~\ref{fig:typ-mul}).
The weight of~\termmult\ is equal to $3`*(\f+\one)+\sum_{b<\f+1}\M$, where the meaning of~\M\ is ``if $b=\f$ then~$0$ else $3\g+1$''.
Thus the execution of the application of~\termmult\ to two integers~\nb\ and~\nb[m] in the \cek\ machine is proportional to~$n`*m$.

\begin{figure*}                                            %
  \centering
  \condinc{\footnotesize
    \fbox{
      \begin{minipage}{.97\linewidth}
        \input{fig-exmult}
      \end{minipage}}
  }{
  \vspace{20pt}
    \input{fig-exmult}
    \vspace{10pt}
  }
  \caption{Typing derivation of~\termmult }
  \label{fig:typ-mul}
\end{figure*}                                              %


%% file: fig-exadd.tex
\begin{flushleft}
  $\A~=~\,\freccia{\NatU{\,\J}}{
    \mtyp[c]{\one}{(\freccia{\NatU{\g}}{
        \NatU{\J+\g}})}}
  $\quad;\quad
  $\C~=~\freccia{\NatU{\H}}{
    \mtyp[c]{\one}{(\freccia{\NatU{\g}}{
        \NatU{\H+\g}})}}
  $\\
  $\,`G~=~\{x:\mtyp{\I}{\A},\ y:\NatU{\H},\ z:\NatU{\g}\}
  $\quad;\quad
  $\fiv~=~\{b,a,c\}$\quad;\quad$\ictx~=~\{b<\f+\one,a<\one,c<\one\}$
\end{flushleft}
\def\defaultHypSeparation{\hskip.2in}
\def\ScoreOverhang{0pt}
\AxiomC{\judg{\fiv}{(\H\geq\one,\ictx)}{
    x:\mtyp{\I}{\A}}{\zero}{x}{\mtyp{\I}{\A}}}
\AxiomC{\judg{\fiv}{(\H\geq\one,\ictx)}{y:\NatU{\H}}{
    \zero}{y}{\NatU{\H}}}
\RightLabel{($p$)}
\UnaryInfC{\judg{\fiv}{(\H\geq\one,\ictx)}{y:\NatU{\H}}{
    \zero}{\pred(y)}{\NatU{\J}}}
\RightLabel{(\textit{App})}
\BinaryInfC{\judg{\fiv}{(\H\geq\one,\ictx)}{x:\mtyp{\I}{\A},
    y:\NatU{\H}}{
    \zero}{x~\pred(y)}{\mtyp[c]{\one}{(\freccia{
        \NatU{\g}}{\NatU{\J+\g}})}}}
\DisplayProof
\begin{flushright}
  \AxiomC{\qquad$\vdots$}
  \AxiomC{\judg{\fiv}{(\H\geq\one,\ictx)}{z:\NatU{\g}}{
      \zero}{z}{\NatU{\g}}}
  \LeftLabel{(\textit{App})}
  \BinaryInfC{\judg{\fiv}{(\H\geq\one,\ictx)}{`G}{\zero}{
      x~\pred(y)\,z}{\NatU{\J+\g}}}
  \LeftLabel{($s$)}
  \UnaryInfC{\judg{\fiv}{(\H\geq\one,\ictx)}{`G}{\zero}{
      \suc(x~\pred(y)\,z)}{\NatU{\H+\g}}}
  \DisplayProof
\end{flushright}
\AxiomC{\judg{\fiv}{\ictx}{y:\NatU{\H}}{\zero}{y}{\NatU{\H}}}
\AxiomC{\judg{\fiv}{(\H\leq\zero,\ictx)}{`G}{\zero}{z}{\NatU{\H+\g}}}
\AxiomC{\vdots\quad}
\RightLabel{(\textit{If})}
\TrinaryInfC{\judg{\fiv}{\ictx}{`G}{\zero}{
    \ifz{y}{z}{\suc(x~\pred(y)\,z)}}{\NatU{\H+\g}}}
\RightLabel{($\multimap$)}
\UnaryInfC{$
  \begin{array}{r}
    \judg{(b,a)}{(b<\f+\one,a<\one)}{(x:\mtyp{\I}{\A};y:\NatU{\H})}{\one}{
      `lz.\ifz{y}{z}{\suc(x~\pred(y)\,z)}}\quad\\
    \mtyp[c]{\one}{(\freccia{\NatU{\g}}{\NatU{\H+\g}})}
  \end{array}
  $}
\RightLabel{($\multimap$)}
\UnaryInfC{$
  \judg{b}{b<\f+\one}{x:\mtyp{\I}{\A}}{\one+\one}{
    `lyz.\ifz{y}{z}{\suc(x~\pred(y)\,z)}}{\mtyp{\one}{\C}}
  $}
\AxiomC{$\hspace{-50pt}
  \ijudg{b}{b<\f+\one}{\C\isubst[b]{b+\one}\equiv\A}
  $}
\LeftLabel{(\textit{Fix})}
\BinaryInfC{\ejudg{}{\f+\one+\sum_{b<\f+\one}(\one+\one)}{\termadd}{
    \mtyp{\one}{\freccia{\NatU{\f}}{
        \mtyp[c]{\one}{(\freccia{\NatU{\g}}{
            \NatU{\f+\g}})}}}
  }}
\DisplayProof


%% file: fig-exmult.tex
\begin{flushleft}
  $(\bigstar):
  \judg{\fiv}{(\H\geq\one,\ictx)}{\emptyset}{3`*(\g+\one)}{
    \termadd}{\mtyp{\one}{\freccia{\NatU{\g}}{
        \mtyp[c]{\one}{(\freccia{\NatU{\J`*\g}}{
            \NatU{\g+\J`*\g}})}}}}
  $
  \\ 
  \vspace{5pt}
  $\A~=~\,\freccia{\NatU{\,\J}}{
    \mtyp[c]{\one}{(\freccia{\NatU{\g}}{
        \NatU{\J`*\g}})}}
  $\quad;\quad
  $\C~=~\freccia{\NatU{\H}}{
    \mtyp[c]{\one}{(\freccia{\NatU{\g}}{
        \NatU{\H`*\g}})}}
  $\\
  $\,`G~=~\{x:\mtyp{\I}{\A},\ y:\NatU{\H},\ z:\NatU{\g}\}
  $\quad;\quad
  $\fiv~=~\{b,a,c\}$\quad;\quad$\ictx~=~\{b<\f+\one,a<\one,c<\one\}$
\end{flushleft}
\def\defaultHypSeparation{\hskip 5pt}
\def\ScoreOverhang{2pt}
\begin{flushright}
  \AxiomC{\judg{\fiv}{(\H\geq\one,\ictx)}{
      x:\mtyp{\I}{\A}}{\zero}{x}{\mtyp{\I}{\A}}}
  \AxiomC{\judg{\fiv}{(\H\geq\one,\ictx)}{y:\NatU{\H}}{
      \zero}{y}{\NatU{\H}}}
  \RightLabel{($p$)}
  \UnaryInfC{\judg{\fiv}{(\H\geq\one,\ictx)}{y:\NatU{\H}}{
      \zero}{\pred(y)}{\NatU{\J}}}
  \RightLabel{(\textit{App})}
  \BinaryInfC{\judg{\fiv}{(\H\geq\one,\ictx)}{x:\mtyp{\I}{\A},
      y:\NatU{\H}}{
      \zero}{x~\pred(y)}{\mtyp[c]{\one}{(\freccia{
          \NatU{\g}}{\NatU{\J`*\g}})}}}
  \DisplayProof
\end{flushright}

\AxiomC{$(\bigstar)\qquad$}
\AxiomC{\judg{\fiv}{(\H\geq\one,\ictx)}{z:\NatU{\g}}{\zero}{z}{\NatU{\g}}}
\BinaryInfC{\judg{\fiv}{(\H\geq\one,\ictx)}{`G}{3`*(\g+1)}{
    \termadd~z}{\mtyp[c]{\one}{(\freccia{\NatU{\J`*\g}}{
        \NatU{\g+\J`*\g}})}}}
\AxiomC{~$\vdots$~}
\AxiomC{\judg{\fiv}{(\H\geq\one,\ictx)}{z:\NatU{\g}}{
    \zero}{z}{\NatU{\g}}}
\RightLabel{(\!\textit{App})}
\BinaryInfC{\judg{\fiv}{(\H\geq\one,\ictx)}{`G}{\zero}{
    x~\pred(y)\,z}{\NatU{\J`*\g}}}

\RightLabel{(\textit{App})}
\BinaryInfC{\judg{\fiv}{(\H\geq\one,\ictx)}{`G}{3`*(\g+1)}{
    \termadd~z~(x\,\pred(y)\,z)}{\NatU{\H`*\g}}}
\DisplayProof

\AxiomC{\judg{\fiv}{\ictx}{y:\NatU{\H}}{\zero}{y}{\NatU{\H}}}
\AxiomC{\judg{\fiv}{(\H\leq\zero,\ictx)}{`G}{\zero}{\nb[0]}{\NatU{\H`*\g}}}
\AxiomC{\qquad\vdots\qquad\qquad}
\RightLabel{(\textit{If})}
\TrinaryInfC{\judg{\fiv}{\ictx}{`G}{\M}{
    \ifz{y}{\nb[0]}{\termadd~(x~\pred(y)\,z)~z}}{\NatU{\H`*\g}}}
\RightLabel{($\multimap$)}
\UnaryInfC{$
  \begin{array}{r}
    \judg{(b,a)}{(b<\f+\one,a<\one)}{(x:\mtyp{\I}{\A};y:\NatU{\H})}{\one+\M}{
      `lz.\ifz{y}{\nb[0]}{\termadd~(x~\pred(y)\,z)~z}}{\qquad}\\
    \mtyp[c]{\one}{(\freccia{\NatU{\g}}{\NatU{\H`*\g}})}
  \end{array}
  $}
\RightLabel{($\multimap$)}
\UnaryInfC{$
  \judg{b}{b<\f+\one}{x:\mtyp{\I}{\A}}{\one+\one+\M}{
    `lyz.\ifz{y}{\nb[0]}{\termadd~(x~\pred(y)\,z)~z}}{\mtyp{\one}{\C}}
  $}
\RightLabel{(\textit{Fix})}
\UnaryInfC{\ejudg{}{\f+\one+\sum_{b<\f+\one}(\one+\one+\M)}{\termmult}{
    \mtyp{\one}{\freccia{\NatU{\f}}{
        \mtyp[c]{\one}{(\freccia{\NatU{\g}}{
            \NatU{\f`*\g}})}}}
  }}
\DisplayProof


%% file: metatheory.tex
\section{The Metatheory of \dlpcfv}
\label{sec:meta}
In this section, some metatheoretical results about
\dlpcfv\ will be presented.
More specifically, type derivations are shown to be modifiable in many different ways, all of them leaving the underlying term unaltered. These manipulations, described in Section~\ref{sec:manipul}, form a basic toolkit which is essential to achieve the main
results of this paper, namely intentional soundness and
completeness (which are presented in Section~\ref{sec:sound} and
Section~\ref{sec:complete}).
Types are preserved by call-by-value reduction, as proved in Section~\ref{sec:subjred}.

\begin{figure*}                                            %
  \centering
  \condinc{\footnotesize
    \fbox{
      \begin{minipage}{.97\linewidth}
        \input{fig-stacktype}
      \end{minipage}}
  }{
    \input{fig-stacktype}
    \vspace{10pt}
  }
  \caption{\dlpcfv: Lifting Typing to Stacks}
  \label{fig:acc-stack}
\end{figure*}                                              %

\subsection{Manipulating Type Derivations}
\label{sec:manipul}
First of all, the constraints $\ictx$ in index, subtyping 
and typing judgements can be made stronger without altering the rest:
\begin{lemma}[Strengthening]
  \label{lem:strength}
  If $\ijudg{\fiv}{\ictxvar}{\ictx}$, then the following implications hold:
  \begin{varenumerate}
    \item\label{point:strfirst}
      If $\ijudg{\fiv}{\ictx}{\I\leq\J}$, then $\ijudg{\fiv}{\ictxvar}{\I\leq\J}$;
    \item\label{point:strsecond}
      If $\sjudg{\fiv}{\ictx}{`s`<`t}$, then $\sjudg{\fiv}{\ictxvar}{`s`<`t}$;
    \item\label{point:strthird}
      If $\judg{\fiv}{\ictx}{`G}{\I}{\t}{`s}$, then $\judg{\fiv}{\ictxvar}{`G}{\I}{\t}{`s}$.
  \end{varenumerate}
\end{lemma}
\begin{proof}
  Point~\ref{point:strfirst}. is a trivial consequence of transitivity of implication in logic.
  Point~\ref{point:strsecond}. can be proved by induction on the structure of the proof of $\sjudg{\fiv}{\ictx}{`s`<`t}$, using point~\ref{point:strfirst}.
  Point~\ref{point:strthird}. can be proved by induction on a proof of $\judg{\fiv}{\ictx}{`G}{\I}{\t}{`s}$, using points~\ref{point:strfirst} and \ref{point:strthird}. 
\end{proof}
Strengthening is quite intuitive:
whatever appears on the right of $\vdash_\ep$ should hold for all values of the variables in $\fiv$ satisfying $\ictx$, so strengthening corresponds to making the judgement weaker. 

Fresh term variables can be added to the context $`G$, leaving the rest of the judgement unchanged:
\begin{lemma}[Context Weakening]\label{lem:ctx-weak}
  $\judg{\fiv}{\ictx}{`G}{\I}{\t}{`t}$ implies
  $\judg{\fiv}{\ictx}{`G,`D}{\I}{\t}{`t}$.
\end{lemma}
\begin{proof}
Again, this is an induction on the structure of a derivation for 
$\judg{\fiv}{\ictx}{`G}{\I}{\t}{`t}$.
\end{proof}
Another useful transformation on type derivations consists in substituting index variables for defined index terms.
\begin{lemma}[Index Substitution]
  \label{lem:isubst}
  If $\ijudg{\fiv}{\ictx}{\I\cvg}$, then the following
  implications hold:
  \begin{varenumerate}
    \item \label{it:isubs-ijudg}
      If $\ijudg{(a,\fiv)}{\ictx,\ictxvar}{\J\leq\K}$, then
      \condinc{}{

      \hfil} 
    $\ijudg{\fiv}{\ictx,\ictxvar\isubst{\I}}{\J\isubst{\I}\leq\K\isubst{\I}}$~;
    \item \label{it:isubs-sjudg}
      If $\sjudg{(a,\fiv)}{\ictx,\ictxvar}{`s`<`t}$, then       \condinc{}{

      \hfil} $\sjudg{\fiv}{\ictx,\ictxvar\isubst{\I}}{`s\isubst{\I}`<`t\isubst{\I}}$~;
    \item \label{it:isubs-judg}
      If $\judg{(a,\fiv)}{\ictx,\ictxvar}{`G}{\J}{\t}{`s}$, then       \condinc{}{

      \hfil} $\judg{\fiv}{\ictx,\ictxvar\isubst{\I}}{`G\isubst{\I}}{\J\isubst{\I}}{\t}{`s\isubst{\I}}$~.
  \end{varenumerate}
\end{lemma}
\begin{proof}
  \begin{varenumerate}
  \item Assume that \ijudg{\fiv}{\ictx}{\I\cvg} and \ijudg{(a,\fiv)}{\ictx,\ictxvar}{\J\leq\K}, and let~$`r$ be an assignment satisfying $\ictx,\ictxvar\isubst{\I}$.
    In particular, $`r$ satisfies~\ictx, thus~\itp{\I} is defined, say equal to~$n$.
    For any index~\H, $\itp{\H\isubst{\I}}=\itp[`r,a\mapsto n]{\H}$.
    Hence $(`r,a\mapsto n)$ satisfies $\ictx,\ictxvar$, and then it also satisfies $\J\leq\K$.
    So $\itp{\J\isubst{\I}}=\itp[`r,a\mapsto n]{\J}\leq\itp[`r,a\mapsto n]{\K}=\itp{\K\isubst{\I}}$, and~$`r$ satisfies $\J\isubst{\I}\leq\K\isubst{\I}$.
    Thus $\ijudg{\fiv}{\ictx,\ictxvar\isubst{\I}}{\J\isubst{\I}\leq\K\isubst{\I}}$.
  \item By induction on the subtyping derivation, using~\ref{it:isubs-ijudg}.
  \item By induction on the typing derivation, using~\ref{it:isubs-ijudg} and \ref{it:isubs-sjudg}.
    \vspace{-1.5em}
  \end{varenumerate}
\end{proof}
Observe that the only hypothesis is that $\ijudg{\fiv}{\ictx}{\I\cvg}$ (definition in Section~\ref{sec:index}): 
we do not require $\I$ to be a value of $a$ that satisfies $\ictxvar$.
If it does not the constraints in $\ictx,\ictxvar\isubst{\I}$ become inconsistent, and the obtained judgements are vacuous.

\subsection{Subject Reduction}
\label{sec:subjred}
What we want to prove in this subsection is the following result:
\begin{proposition}[Subject Reduction]
  \label{prop:sr}
  If $\t\tov\u$ and $\judg{\fiv}{\ictx}{`0}{\M}{\t}{`t}$,
  then $\judg{\fiv}{\ictx}{`0}{\M}{\u}{`t}$.
\end{proposition}
Subject Reduction can be proved in a standard way, by going through
a Substitution Lemma, which only needs to be proved when the term 
being substituted is a \emph{value}. Preliminary to the Substitution Lemma
are two auxiliary results stating that derivations giving types to
values can, if certain conditions hold, be split into two,
or put in parametric form:
\begin{lemma}[Splitting]
  \label{lem:cut-sum}
  If $\judg{\fiv}{\ictx}{`G}{\M}{\v}{`t_1\uplus`t_2}$,
  then there exist two indexes $\N_1,\N_2$, and two contexts $`G_1,`G_2$, such that
  $\judg{\fiv}{\ictx}{`G_i}{\N_i}{\v}{`t_i}$, and
  $\ijudg{\fiv}{\ictx}{\N_1+\N_2\leq\M}$ and
  $\sjudg{\fiv}{\ictx}{`G`<`G_1\uplus`G_2}$.
\end{lemma}
\begin{proof}
  If~\v\ is a primitive integer~\nb, the result is trivial as the only possible decomposition of a type for integers is $\Nat{\I}{\J}=\Nat{\I}{\J}\uplus\Nat{\I}{\J}$.
  
  If $\v=`lx.\t$, then its typing judgement derives from
  \begin{align}
    \label{eq:hyp-abstr-one}
    \judg{(a,\fiv)}{(a<\I,\ictx)}{`D,x:`s&}{\K}{\t}{`t} \\
    \label{eq:hyp-abstr-two}
    \sjudg{\fiv}{\ictx&}{`G`<\sum_{a<\I}`D}
  \end{align}
  with $`t_1\uplus`t_2=\mtyp{\I}{\freccia{`s}{`t}}$ and $\M=\I+\sum_{a<\I}\K$.
  Hence $\I=\I_1+\I_2$, and $`t_1=\mtyp{\I_1}{\freccia{`s}{`t}}$, and $`t_2=\mtyp{\I_2}{\freccia{`s\isubst{\I_1+a}}{`t\isubst{\I_1+a}}}$.
  Since \ijudg{(a,\fiv)}{(a<\I_1,\ictx)}{(a<\I,\ictx)}, we can strength the hypothesis in~\eqref{eq:hyp-abstr-one} by Lemma~\ref{lem:strength} and derive
  \begin{center}
    \AxiomC{\judg{(a,\fiv)}{(a<\I_1,\ictx)}{`D,x:`s}{\K}{\t}{`t}}
    \UnaryInfC{\judg{\fiv}{\ictx}{\sum_{a<\I_1}`D}{\I_1+\sum_{a<\I_1}\K}{
        `lx.\t}{\mtyp{\I_1}{\freccia{`s}{`t}}}}
    \DisplayProof
  \end{center}
  On the other hand, we can substitute~$a$ with~$a+\I_1$ in~\eqref{eq:hyp-abstr-one} by Lemma~\ref{lem:isubst}, and derive
  \begin{center}
    \AxiomC{\judg{(a,\fiv)}{(a<\I_2,\ictx)}{`D\isubst{a+\I_1},
        x:`s\isubst{a+\I_1}}{\K\isubst{a+\I_1}}{\t}{`t\isubst{a+\I_1}}}
    \UnaryInfC{\judg{\fiv}{\ictx}{\sum_{a<\I_2}`D\isubst{a+\I_1}}{
        \I_2+\sum_{a<\I_2}\K\isubst{a+\I_1}}{`lx.\t}{
        \mtyp{\I_2}{\freccia{`s\isubst{a+\I_1}}{`t\isubst{a+\I_1}}}}}
    \DisplayProof
  \end{center}
  Hence we can conclude with $`G_1=\sum_{a<\I_1}`D$, $`G_2=\sum_{a<\I_2}`D\isubst{a+\I_1}$, $\N_1=\I_1+\sum_{a<\I_1}\K$ and $\N_2=\I_2+\sum_{a<\I_2}\K\isubst{a+\I_1}$.
  
  Now, if $\v=\fix{\t}$, then its typing judgement derives from
  \begin{align}
    \label{eq:hyp-fix1}
    \judg{(b,\fiv)}{(b<\H,\ictx)}{`D,x:\mtyp{\I}{\A}&}{\J}{\t}{\mtyp{1}{\B}}\\
    \label{eq:hyp-fix2}
    \ijudg{\fiv}{\ictx&}{\textstyle\H\geq\fc{0}{\K}{\I}} \\
    \label{eq:hyp-fix3}
    \sjudg{(a,b,\fiv)}{(a<\I,b<\H,\ictx)&}{
      \B\isubst{\zero}\isubst[b]{\textstyle\fc{b+1}{a}{\I}+b+1}`<\A}\\
    \label{eq:hyp-fix4}
    \sjudg{(a,\fiv)}{(a<\K,\ictx)&}{
      \B\isubst{\zero}\isubst[b]{\textstyle\fc{0}{a}{\I}}`<\C}\\
    \label{eq:hyp-fix5}
    \sjudg{\fiv}{\ictx&}{\textstyle`G`<\sum_{b<\H}`D} 
  \end{align}
  with $`t_1\uplus`t_2=\mtyp{\K}{\C}$, and $\M=\H+\sum_{b<\H}\J$.
  Hence $\K=K_1+\K_2$, with $`t_1=\mtyp{\K_1}{\C}$, and $`t_2=\mtyp{\K_2}{\C\isubst{a+\K_1}}$.
  Let $\H_1=\fc{\zero}{\K_1}{\I}$ and $\H_2=\fc{\H_1}{\K_2}{\I}$.
  Then $\H_1+\H_2=\fc{\zero}{\K}{\I}$, and~$\H_2$ is also equal to $\fc{\zero}{\K_2}{\I\isubst[b]{\H_1+b}}$.
  Just like the previous case, we can strengthen the hypothesis in~\eqref{eq:hyp-fix1}, \eqref{eq:hyp-fix3} and~\eqref{eq:hyp-fix4} and derive
  \begin{center}
    \AxiomC{$
      \begin{array}{r@{\,}l}
        \judg{(b,\fiv)}{(b<\H_1,\ictx)}{`D,x:\mtyp{\I}{\A}&}{\J}{\t}{\mtyp{1}{\B}}\\
        \sjudg{(a,b,\fiv)}{(a<\I,b<\H_1,\ictx)&}{
        \B\isubst{\zero}\isubst[b]{\fc{b+1}{a}{\I}+b+1}`<\A}\\
        \sjudg{(a,\fiv)}{(a<\K_1,\ictx)&}{
        \B\isubst{\zero}\isubst[b]{\fc{0}{a}{\I}}`<\C} \\
      \end{array}
      $}
    \UnaryInfC{\judg{\fiv}{\ictx}{\sum_{b<\H_1}`D}{\H_1+\sum_{b<\H_1}\J}{\fix{\t}}{\mtyp{\K_1}{\C}}}
    \DisplayProof
  \end{center}
  Moreover, if we substitute~$b$ with $b+\H_1$ in~\eqref{eq:hyp-fix3} and we strengthen the constraints (since~\eqref{eq:hyp-fix2} implies \ijudg{\fiv}{\ictx,b<\H_2}{\ictx,b+\H_1<\H}), we get
  \begin{displaymath}
    \sjudg{(a,b,\fiv)}{(a<\I,b<\H_2,\ictx)}{
      \B\isubst{\zero}\isubst[b]{\textstyle\fc{b+1}{a}{\I}+b+1}
      \isubst[b]{\H_1+b}`<\A\isubst[b]{\H_1+b}}.
  \end{displaymath}
  But $\big(\fc{b+1}{a}{\I}+b+1\big)\isubst[b]{\H_1+b} = \fc{\H_1+b+1}{a}{\I}+\H_1+b+1$ and $\fc{\H_1+b+1}{a}{\I} = \fc{b+1}{a}{(\I\isubst[b]{\H_1+b})}$.
  Hence $\B\isubst{\zero}\isubst[b]{\textstyle\fc{b+1}{a}{\I}+b+1}\isubst[b]{\H_1+b} = 
  \B\isubst[b]{\H_1+b}\isubst{\zero}\isubst[b]{\textstyle\fc{b+1}{a}{(\I\isubst[b]{\H_1+b})}+b+1}$.
  \\
  In the same way we can substitute~$a$ with $a+\K_1$ in~\eqref{eq:hyp-fix4}:
  \begin{displaymath}
    \sjudg{(a,\fiv)}{(a<\K_2,\ictx)}{
      \B\isubst{\zero}\isubst[b]{\textstyle\fc{0}{a+\K_1}{\I}}`<\C\isubst{a+\K_1}}
  \end{displaymath}
  But $\fc{0}{a+\K_1}{\I} = \H_1+\fc{\H_1}{a}{\I} = \H_1+\fc{\zero}{a}{\I\isubst[b]{\H_1+b}}$, and so $\B\isubst{\zero}\isubst[b]{\fc{0}{a+\K_1}{\I}}$ is equivalent to $\B\isubst[b]{\H_1+b}\isubst{\zero}\isubst[b]{\fc{0}{a}{\I\isubst[b]{\H_1+b}}}$.
  Finally, by substituting also~$b$ with $b+\H_1$ in (\ref{eq:hyp-fix1}) we can derive
  \begin{center}
    \AxiomC{$
      \begin{array}{c}
        \judg{(b,\fiv)}{(b<\H_2,\ictx)}{`D\isubst[b]{\H_1+b},
          x:(\mtyp{\I}{\A})\isubst[b]{\H_1+b}}{
          \J\isubst[b]{\H_1+b}}{\t}{\mtyp{1}{\B\isubst[b]{\H_1+b}}}\\
        \sjudg{(a,b,\fiv)}{(a<\I,b<\H_2,\ictx)}{\B\isubst[b]{\H_1+b}\isubst{\zero}\isubst[b]{\fc{b+1}{a}{(\I\isubst[b]{\H_1+b})}+b+1} `< \A\isubst[b]{\H_1+b}}\\
        \sjudg{(a,\fiv)}{(a<\K_2,\ictx)}{
        \B\isubst[b]{\H_1+b}\isubst{\zero}\isubst[b]{\fc{0}{a}{\I\isubst[b]{\H_1+b}}}`<
        \C\isubst{a+\K_1}} \\
      \end{array}
      $}
    \UnaryInfC{\judg{\fiv}{\ictx}{\sum_{b<\H_2}`D\isubst[b]{\H_1+b}}{
        \H_2+\sum_{b<\H_2}\J\isubst[b]{\H_1+b}}{\fix{\t}}{\mtyp{\K_2}{\C\isubst{a+\K_1}}}}
    \DisplayProof
  \end{center}
  So we can conclude with $`G_1=\sum_{a<\H_1}`D$, $`G_2=\sum_{a<\H_2}`D\isubst{a+\H_1}$, $\N_1=\H_1+\sum_{a<\H_1}\J$ and $\N_2=\H_2+\sum_{a<\H_2}\J\isubst{a+\H_1}$.
\end{proof}

\begin{lemma}[Parametric Splitting]
  \label{lem:cut-boundedsum}
  If $\judg{\fiv}{\ictx}{`G}{\M}{\v}{\sum_{c<\J}`s}$
  is derivable, then there exist an index~\N\ and a context~$`D$ such that one can derive $\judg{c,\fiv}{c<\J,\ictx}{`D}{\N}{\v}{`s}$, and
  $\ijudg{\fiv}{\ictx}{\sum_{c<\J}\N\leq\M}$
  and $\sjudg{\fiv}{\ictx}{`G`<\sum_{c<\J}`D}$.
\end{lemma}
\begin{proof}
  The proof uses the same technique as for Lemma~\ref{lem:cut-sum}.
  If~\v\ is a lambda abstraction or a fixpoint, then $\sum_{c<\J}`s$ is on the form $\mtyp{\sum_{c<\J}\L}{\C}$, where $\mtyp{\L}{\C\isubst{a+\sum_{c'<c}\L\isubst[c]{c'}}}=`s$.
  Then the result also follows from Strengthening (Lemma~\ref{lem:strength}) and Index Substitution (Lemma~\ref{lem:isubst}):
  for the lambda abstraction, substitute~$a$ with $a+\sum_{c'<c}\L\isubst[c]{c'}$ in~(\ref{eq:hyp-abstr-one}).
  For the fixpoint consider the index~$\H'$ satisfying the equations $\H'\isubst[c]{\zero}=\fc{\zero}{\L\isubst[c]{\zero}}{\I}$ and $\H'\isubst[c]{\inb[i]+\one}=\fc{\H\isubst[c]{\inb[i]}}{\L\isubst[c]{\inb[i]+\one}}{\I}$.
  Then substitute~$b$ with $b+\sum_{c'<c}\H'\isubst[c]{c'}$ (and add the constraint $c<\J$ in the context) in~\eqref{eq:hyp-fix1} and~\eqref{eq:hyp-fix3}, and substitute~$a$ with $a+\sum_{c'<c}\L\isubst[c]{c'}$ in~\eqref{eq:hyp-fix4} to derive the result.  
\end{proof}

One can easily realise \emph{why} these results are crucial for subject
reduction: whenever the substituted value flows through a type derivation,
there are various places where its type changes, namely when it reaches
instances of the typing rules $(\mathit{App})$, $(\multimap)$, $(\mathit{If})$ and $(\mathit{Rec})$:
in all these cases the type derivation for the value must be modified,
and the splitting lemmas certify that this is possible. We can this
way reach the key intermediate result:
\begin{lemma}[Substitution]
  \label{lem-termsubst}
  If
  $\judg{\fiv}{\ictx}{`G,x:`s}{\M}{\t}{`t}$
  and $\judg{\fiv}{\ictx}{`0}{\N}{\v}{`s}$
  are both derivable, then there is an index~$\K$ such that
  $\judg{\fiv}{\ictx}{`G}{\K}{\t\tsubst{\v}}{`t}$
  and $\ijudg{\fiv}{\ictx}{\K\leq \M+\N}$.
\end{lemma}
\begin{proof}
  The proof goes by induction on the derivation of the judgement \judg{\fiv}{\ictx}{`G,x:`s}{\M}{\t}{`t}, making intense use of Lemma~\ref{lem:cut-sum} and Lemma~\ref{lem:cut-boundedsum}.
\end{proof}
Given Lemma~\ref{lem-termsubst}, proving Proposition~\ref{prop:sr} is routine: the only
two nontrivial cases are those where the fired redex is a $\beta$-redex or the
unfolding of a recursively-defined function, and both consist in a substitution.
Observe how Subject Reduction already embeds a form of \emph{extensional} soundness
for \dlpcfv, since types are preserved by reduction.
As an example, if one builds a type derivation for $\ejudg{}{\I}{\t}{\Nat{2}{7}}$, then the normal form of $\t$ (if it exists) is guaranteed to be a constant between $2$ and $7$.
Observe, on the other hand, than nothing is known about
the \emph{complexity} of the underlying computational process yet, since the weight $\I$ does not necessarily decrease along reduction.
This is the topic of the following section.

\subsection{Intentional Soundness}
\label{sec:sound}
\begin{figure*}[htbp]                                       %
  \centering
  \condinc{
    \framebox{
      \input{fig-stacksize}}
  }{
    \input{fig-stacksize}
  }
  \caption{Size of processes}
  \label{fig:size-proc}
\end{figure*}                                               %
In this section, we prove the following result:
\begin{theorem}[Intensional soundness]
  \label{theo:isnd}
  For any term~\t, if $$\ejudg{}{\H}{\t}{\Nat{\I}{\J}}$$
  then $\t\conv\nb[m]$ where $n\leq\ts{\t}`.(\itp[]{\H}+1)$
  and $\itp[]{\I}\leq m\leq\itp[]{\J}$~.
\end{theorem}
Roughly speaking, this means that \dlpcfv\ also gives us some
sensible information about the time complexity of evaluating
typable \PCF\ programs. The path towards Theorem~\ref{theo:isnd}
is not too short: it is necessary to lift \dlpcfv\ to a type
system for closures, environments and processes, as defined in 
Section~\ref{sec:cek}. Actually, the type system can be
easily generalised to closures by the rule below:
$$
\infer[]
{\cjudg{\fiv}{\ictx}{\K+\sum_{1\leq i\leq n}\J_i}{
    \clo[\{x_1\mapsto\clov_1;\cdots;x_n\mapsto\clov_n\}]{t}}{`t}}
{
  \begin{array}{c}
    \judg{\fiv}{\ictx}{x_1:`s_1,\dots,x_n:`s_n}{\K}{\t}{`t}\\
    \cjudg{\fiv}{\ictx}{\J_i}{\clov_i}{`s_i}\\
  \end{array}
}
$$
Lifting everything to stacks, on the other hand, requires more
work, see Figure~\ref{fig:acc-stack}.
We say that a stack~\stack\ is \fifi-\emph{acceptable} for~$`s$ with type~$`t$ with cost~\I\ (notation: \stjudg{\fiv}{\ictx}{\I}{\stack}{`s}{`t}) when it interacts well with closures of type~$`s$ to product a process of type~$`t$.
Indeed, a \emph{process} can be typed as follows:
$$
\infer[]
{\pjudg{\fiv}{\ictx}{\J+\K}{\proc{\cloc}{\stack}}{`t}}
{
  \begin{array}{c}
    \stjudg{\fiv}{\ictx}{\J}{\stack}{`s}{`t}\\
    \cjudg{\fiv}{\ictx}{\K}{\cloc}{`s}\\
  \end{array}
}  
$$
This way, also the notion of weight has been lifted to processes,
with the hope of being able to show that it strictly decreases
at every evaluation step. Apparently, this cannot be achieved in full:
sometimes the weight of a process does not change, but in that case
another parameter is guaranteed to decrease, namely the process \emph{size}.
The size $\sip{\proc{\cloc}{\stack}}$ of 
$\proc{\cloc}{\stack}$, is defined as 
$\sic{\cloc}+\sis{\stack}$, where:
\begin{varitemize}
\item
  The size $\sic{\cloc}$ of a closure $\clo{\t}$ is the
  \emph{multiplicative} size of $\t$ (\cf Section~\ref{sec:prg}).
\item
  The size of $\sic{\stack}$ is the sum of the sizes of all closures
  appearing in $\stack$ plus the number of occurrences of
  symbols (different from $\emptystack$ and $\mathsf{fun}$) in $\stack$.
\end{varitemize} 
  The formal definition of $\sip{\proc{\cloc}{\stack}}$ is given in Figure~\ref{fig:size-proc}.

\condinc{
The size of a process decreases by any evaluation steps, except the two ones performing a substitution~\eqref{eq:fun-rul1} and~\eqref{eq:fun-rul2}.
However, these two reduction rules make the \textit{weight} of a process decrease, as formalised by the following proposition.
By the way, these are the cases in which a box is opened up in the underlying linear logic proof.
}{}
\begin{proposition}[Weighted Subject Reduction]
  \label{prop:decreasing}
  Assume $\process\tocek\R$ and $\pjudg{\fiv}{\ictx}{\I}{\process}{`t}$. 
  Then $\pjudg{\fiv}{\ictx}{\J}{\R}{`t}$ and
  \begin{itemize}
  \item either $\ijudg{\fiv}{\ictx}{\I=\J}$ and $\sip{\process}>\sip{\R}$,
  \item 
    or $\ijudg{\fiv}{\ictx}{\I>\J}$ and
    $\sip{\process}+\ts{\s}>\sip{\R}$, where
    $\s$ is a term appearing in $\process$.
  \end{itemize}
\end{proposition}
\begin{proof}
  \begin{varenumerate}
  \item
    If $\process\tocek\process[R]$ with a non substitution rule (any rule of Figure~\ref{fig:cek-val} or Figure~\ref{fig:cek-clo} except \eqref{eq:fun-rul1} and~\eqref{eq:fun-rul2}), then it is easy to check that $\sip{\process}>\sip{\process[R]}$.
    Moreover, in all these cases \process\ and \process[R] have the same type and the same weight.
    We detail some cases:
    \begin{varitemize}
    \item If $\process~=~\proc{\clov}{\stackarg{\cloc}}\quad\tocek\quad
      \proc{\cloc}{\stackfun{\clov}}~=~\process[R]$,\quad
      then the typing of~\process\ derives from
      \begin{center}
        \AxiomC{$
          \begin{array}{r@{\,}l}
            \cjudg{\fiv}{\ictx&}{\H}{\cloc}{`s_0\isubst{\zero}}\\
            \stjudg{\fiv}{\ictx&}{\L}{\stack'}{`t_0\isubst{\zero}}{`t}\\
            \sjudg{\fiv}{\ictx&}{`s`<\mtyp{\one}{(\freccia{`s_0}{`t_0})}}\\
            \ijudg{\fiv}{\ictx&}{\J=\H+\L}
          \end{array}      
          $}
        \UnaryInfC{\stjudg{\fiv}{\ictx}{\J}{\stackarg{\cloc}}{`s}{`t}}
        \AxiomC{$
          \begin{array}{r@{\,}l}
            \cjudg{\fiv}{\ictx&}{\K}{\clov}{`s}\\
            \ijudg{\fiv}{\ictx&}{\I=\J+\K}      
          \end{array}
          $}
        \BinaryInfC{\pjudg{\fiv}{\ictx}{\I}{\proc{\clov}{\stackarg{\cloc}}}{`t}}
        \DisplayProof
      \end{center}
      Hence since subtyping is derivable (Lemma~\ref{lem:subtyp}) we can derive for~\process[R]:
      \begin{center}
        \AxiomC{$
          \begin{array}{r@{\,}l}
            \cjudg{\fiv}{\ictx&}{\K}{\clov}{\mtyp{\one}{(\freccia{`s_0}{`t_0})}}\\
            \stjudg{\fiv}{\ictx&}{\L}{\stack}{`t_0\isubst{\zero}}{`t}
          \end{array}      
          $}
        \UnaryInfC{\stjudg{\fiv}{\ictx}{\L+\K}{\stackfun{\clov}}{`s_0\isubst{\zero}}{`t}}
        \AxiomC{$
          \begin{array}{r@{\,}l}
            \cjudg{\fiv}{\ictx&}{\H}{\cloc}{`s_0\isubst{\zero}}\\
            \ijudg{\fiv}{\ictx&}{\I=\H+\L+\K}      
          \end{array}
          $}
        \BinaryInfC{\pjudg{\fiv}{\ictx}{\I}{\proc{\cloc}{\stackfun{\clov}}}{`t}}
        \DisplayProof
      \end{center}
    \item If $\process=~\proc{\clo{\t\u}}{\stack}~\tocek~ 
      \proc{\clo{\t}}{\stackarg{\clo{\u}}}~=\process[R]$,\quad
      then the typing of~\process\ derives from
      \begin{center}
        \def\defaultHypSeparation{\hskip 0pt}
        \def\ScoreOverhang{0pt}
        \AxiomC{$
          \begin{array}{r@{\,}l}
            \judg{\fiv}{\ictx}{x_1:`m_1,\dots,x_n:`m_n&}{\K}{\t}{\mtyp{\N}{\freccia{`k}{`h}}}\\
            \judg{\fiv}{\ictx}{x_1:`h_1,\dots,x_n:`h_n&}{\H}{\u}{`k\isubst{\zero}}\\
            \sjudg{\fiv}{\ictx&}{`s_i `< `m_i\uplus`h_i}\\
            \ijudg{\fiv}{\ictx&}{\N\geq\one}\\          
            \sjudg{\fiv}{\ictx&}{`h\isubst{\zero} `< `s} 
          \end{array}
          $}
        \UnaryInfC{\judg{\fiv}{\ictx}{x_1:`s_1,\dots,x_n:`s_n}{\H+\K}{\t\u}{`s}}
        \AxiomC{\cjudg{\fiv}{\ictx}{\J_i}{\clov_i}{`s_i}}
        \BinaryInfC{\cjudg{\fiv}{\ictx}{\H+\K+\sum_{i\leq n}\J_i}{\clo{\t\u}}{`s}}
        \AxiomC{$
          \begin{array}{c}
            \stjudg{\fiv}{\ictx}{\J}{\stack}{`s}{`t}\\
            \ijudg{\fiv}{\ictx}{\I=\J+\H+\K+\sum_i\J_i}      
          \end{array}
          $}
        \BinaryInfC{\pjudg{\fiv}{\ictx}{\I}{\proc{\clo{\t\u}}{\stack}}{`t}}
        \DisplayProof
        \def\defaultHypSeparation{\hskip.2in}
        \def\ScoreOverhang{4pt}
      \end{center}
      In particular, since subtyping is derivable, \cjudg{\fiv}{\ictx}{\J_i}{\clov_i}{`m_i\uplus`h_i} for each~$i$.
      By Lemma~\ref{lem:cut-sum} (that can be trivially extended to closures), it means that there are some~$<X_i, \N_i$ such that
      \begin{align*}
        &\cjudg{\fiv}{\ictx}{\M_i}{\clov_i}{`m_i}\\
        &\cjudg{\fiv}{\ictx}{\N_i}{\clov_i}{`h_i}\\
        &\ijudg{\fiv}{\ictx}{\M_i+\N_i=\J_i}
      \end{align*}  
      Hence both these judgements are derivable:
      \begin{center}
        \AxiomC{\judg{\fiv}{\ictx}{x_1:`m_1,\dots,x_n:`m_n}{\K}{\t}{\mtyp{\one}{\freccia{`k}{`h}}}}
        \AxiomC{\cjudg{\fiv}{\ictx}{\M_i}{\clov_i}{`m_i}}
        \BinaryInfC{\cjudg{\fiv}{\ictx}{\K+\sum_{i\leq n}\M_i}{\clo{t}}{\mtyp{\one}{\freccia{`k}{`h}}}}
        \DisplayProof
      \end{center}
      \qquad and\qquad
      \AxiomC{\judg{\fiv}{\ictx}{x_1:`h_1,\dots,x_n:`h_n}{\H}{\u}{`k\isubst{\zero}}}
      \AxiomC{\cjudg{\fiv}{\ictx}{\N_i}{\clov_i}{`h_i}}
      \BinaryInfC{\cjudg{\fiv}{\ictx}{\H+\sum_{i\leq n}\N_i}{\clo{\u}}{`k\isubst{\zero}}}
      \DisplayProof

      Hence we can derive the following typing judgement for~\process[R] (notice that subtyping is derivable for the stacks, with contravariance in the first type):
      \begin{center}
        \def\defaultHypSeparation{\hskip 0pt}
        \def\ScoreOverhang{0pt}
        \AxiomC{$
          \begin{array}{r@{\,}l}
           \cjudg{\fiv}{\ictx&}{\K+\sum_{i\leq n}\M_i}{\clo{t}}{\mtyp{\one}{\freccia{`k}{`h}}}\\
           \ijudg{\fiv}{\ictx&}{\I=\K+\J+\H+\sum_i(\M_i+\N_i)}
          \end{array}
          $}
        \AxiomC{$
          \begin{array}{r@{\,}l}
            \cjudg{\fiv}{\ictx&}{\H+\sum_{i\leq n}\N_i}{\clo{\u}}{`k\isubst{\zero}}\\
            \stjudg{\fiv}{\ictx&}{\J}{\stack}{`h\isubst{\zero}}{`t}\\            
          \end{array}
          $}
        \UnaryInfC{\stjudg{\fiv}{\ictx}{\J+\H+\sum_i\N_i}{\stackarg{\clo{\u}}}{
          \mtyp{\one}{\freccia{`k}{`h}}}{`t}}
        \BinaryInfC{\pjudg{\fiv}{\ictx}{\I}{\proc{\clo{\t}}{\stackarg{\clo{\u}}}}{`t}}
        \DisplayProof
        \def\defaultHypSeparation{\hskip.2in}
        \def\ScoreOverhang{4pt}
      \end{center}
    \end{varitemize}
    \item If $\process\tocek\process[R]$ with a substitution rule...
  \end{varenumerate}
\end{proof}
\condinc{}{
Actually, Proposition~\ref{prop:decreasing} can be proved by carefully
analysing the various cases as for how $\process$ evolves to $\R$, \ie
the rules from Figure~\ref{fig:cek-val} and Figure~\ref{fig:cek-clo}.
Only in two of them the weight decreases, namely the ones
in which $\process$ is in the form $\proc{\clov}{\stackfun{\clov[w]}}$.
By the way, these are the cases in which a box is opened up in the underlying linear logic proof.
}
Splitting and parametric splitting play a crucial role here, once appropriately generalised to value closures.

Given Proposition~\ref{prop:decreasing}, Theorem~\ref{theo:isnd} is
within reach: the natural number $\ts{\s}$ in Proposition~\ref{prop:decreasing} 
cannot be greater than the size of the term $\t$ we start from, since the only ``new''
terms created along reduction are constants in the form $\nb$ (which have null
size). 
\subsection{(Relative) Completeness}
\label{sec:complete}
In this section, we will prove some results about the expressive
power of \dlpcfv, seen as a tool to prove intentional (but also extensional)
properties of \PCF\ terms. Actually, \dlpcfv\ is extremely powerful:
every first-order \PCF\ program computing the function $f:\int\rightarrow\int$
in a number of steps bounded by $g:\int\rightarrow\int$ can be proved
to enjoy these properties by way of \dlpcfv, provided two conditions
are satisfied:
\begin{varitemize}
\item
  On the one hand, the equational program $\ep$ needs to be \emph{universal},
  meaning that every partial recursive function is expressible
  by some index terms. This can be guaranteed, as an example, by the presence
  of a universal program in $\ep$.
\item
  On the other hand, all \emph{true} statements in the form $\ijudg{\fiv}{\ictx}{\I\leq\J}$
  must be ``available'' in the type system for completeness to hold. In other words, one cannot
  assume that those judgements are derived in a given (recursively enumerable) formal
  system, because this would violate G\"odel's Incompleteness Theorem. In fact,
  ours are completeness theorems \emph{relative} to an oracle for the truth of those
  assumptions, which is precisely what happens in Floyd-Hoare logics~\cite{Cook78}.
\end{varitemize}
\paragraph{\PCF\ Typing} The first step towards completeness is quite easy: propositional
type systems in the style of \PCF\ for terms, closures, stacks and processes need to be 
introduced. All of them can be easily obtained by erasing the index information from
\dlpcfv. As an example, the typing rule for the application looks like
$$
\infer[]
{\pcftjudg{`G}{\t\u}{`b}}
{\pcftjudg{`G}{\t}{`a\arr`b} 
 \qquad
 {\pcftjudg{`G}{\u}{`a}}
}
$$
while processes can be typed by the following rule
$$
\infer[]
{\pcfpjudg{\proc{\cloc}{\stack}}{`b}}
{
  \pcfstjudg{\stack}{`a}{`b}
  \qquad
  \pcfpjudg{\cloc}{`a}
}  
$$
Given any type $`s$ (respectively any type derivation $\pcftder$) of \dlpcfv, 
the \PCF\ type (respectively, the \PCF\ type derivation) obtained by
erasing all the index information will be denoted by $\erase{`s}$
(respectively, by $\erase{\pcftder}$). Of course both terms and processes enjoy subject reduction 
theorems with respect to \PCF\ typing, and their proofs are much simpler than those for 
\dlpcfv. As an example, given a type derivation $\pcftder$ for $\pcfpjudg{\process}{\NatPCF}$ (we might write $\pcftder\proves\pcfpjudg{\process}{\NatPCF}$)
and $\process\tocek\R$, a type derivation $\pcftder'$ for $\pcfpjudg{\R}{\NatPCF}$ can be easily built by
manipulating in a standard way $\pcftder$; we write $\pcftder\tocek\pcftder'$.
\paragraph{Weighted Subject Expansion}
The key ingredient for completeness is a dualisation of Weighted Subject Reduction:
\begin{proposition}[Weighted Subject Expansion]
  \label{prop:weisubjexp}
  Suppose that $\pcftder\proves\pcfpjudg{\process}{`a}$,
  that $\pcftder\tocek\pcftder'$, and that
  $\tder'\proves\pjudg{\fiv}{\ictx}{\I}{\R}{`t}$
  where $\erase{\tder'}=\pcftder'$.
  Then there is
  $$\tder\proves\pjudg{\fiv}{\ictx}{\J}{\process}{`t}$$
  with $\erase{\tder}=\pcftder$ and 
  $\ijudg{\fiv}{\ictx}{\J\leq\I+1}$. 
  Moreover, $\tder$ can be effectively computed from $\pcftder$,  
  $\tder'$ and $\pcftder'$.
\end{proposition}
Proving Proposition~\ref{prop:weisubjexp} requires a careful analysis of the
evolution of the \CEKPCF\ machine, similarly to what happened for Weighted
Subject \emph{Reduction}. But while in the latter it is crucial to be
able to (parametrically) \emph{split} type derivations for terms (and thus closures),
here we need to be able to \emph{join} them:
\begin{lemma}[Joining]
  If $\ep$ is universal, then
  \begin{displaymath}
    \left.
      \begin{array}[c]{r@{\,}l}
        \pcftder_i\proves\judg{\fiv}{\ictx}{`G_i&}{\N_i}{\v}{`t_i}\\
        \erase{\pcftder_1}&=\erase{\pcftder_2} \\
        \sjudg{\fiv}{\ictx&}{`G`<`G_1\uplus`G_2} \\
        \sjudg{\fiv}{\ictx&}{`t_1\uplus`t_2`<`t} \\
        \ijudg{\fiv}{\ictx&}{\N_1+\N_2\leq\M}
      \end{array}
    \right\}\implies
    \judg{\fiv}{\ictx}{`G}{\M}{\v}{`t}
  \end{displaymath}
\end{lemma}
\begin{lemma}[Parametric Joining]
  Suppose that $\ep$ is universal.
  Then
  \begin{displaymath}
    \left.
      \begin{array}[c]{r@{\,}l}
        \judg{a,\fiv}{a<\I,\ictx}{`D&}{\N}{\v}{`s} \\
        \sjudg{\fiv}{\ictx&}{`G`<\sum_{a<\I}`D} \\
        \sjudg{\fiv}{\ictx&}{\sum_{a<\I}`s`<`t} \\
        \ijudg{\fiv}{\ictx&}{\sum_{a<\I}\N\leq\M}
      \end{array}
    \right\}\implies
    \judg{\fiv}{\ictx}{`G}{\M}{\v}{`t}
  \end{displaymath}
\end{lemma}
Observe that the Joining Lemma requires the two type derivations to be joined to
have the same \PCF\ ``skeleton''. This is essential, because otherwise it
would not be possible to unify them into one single type derivation.
\paragraph{Completeness for Programs}
We now have all the necessary ingredients to obtain a first completeness
result, namely one about programs (which are terms of type $\NatPCF$).
Suppose that $\t$ is a \PCF\ program such that
$\t\tov^*m$, where $m$ is a natural number. 
By Proposition~\ref{prop:eq}, there is a sequence
of processes
$$
\process_1\tocek\process_2\tocek\ldots\tocek\process_n,
$$
where $\process_1=(\proc{\clo[\emptyset]{\t}}{\emptystack})$
and $\process_n=(\proc{\clo[\emptyset]{m}}{\emptystack})$.
Of course, $\vdash\process_i:\NatPCF$ for every $i$.
For obvious reasons, $\vdash^\ep_{0}\process_{n}:\NatU{m}$.
Moreover, by Weighted Subject Expansion, we can derive
each of $\vdash^\ep_{\I_i}\process_{i}:\NatU{m}$, until
we reach $\vdash^\ep_{\I_1}\process_{1}:\NatU{m}$,
where $\I_1\leq n$ (see Figure~\ref{fig:complprog}
for a graphical representation of the above argument).
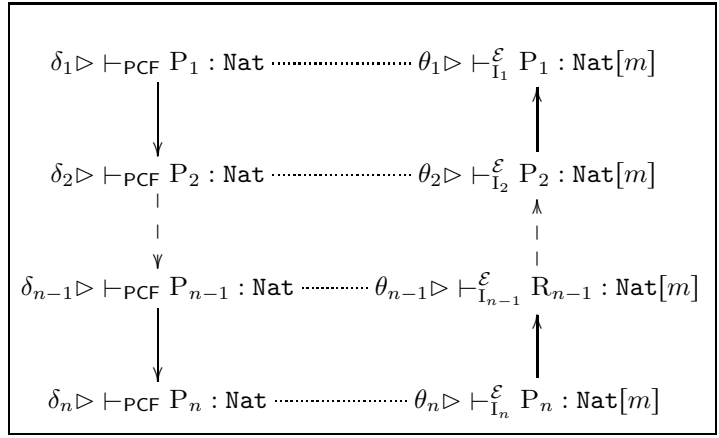
\begin{figure}%
  \centering
  \condinc{
    \fbox{
      \begin{minipage}{.6\textwidth}
        $$\hspace{-5pt}
        \xymatrix{
          \pcftder_1\proves\pcfpjudg{\process_1}{\NatPCF}\ar[d]\ar@{.}[r] &
          \tder_1\proves\ejudg{}{\I_1}{\process_1}{\NatU{m}} \\
          \pcftder_2\proves\pcfpjudg{\process_2}{\NatPCF}
          \ar@{-->}[d]\ar@{.}[r] &
          \tder_2\proves\ejudg{}{\I_2}{\process_2}{\NatU{m}}\ar[u] \\
          \pcftder_{n-1}\proves\pcfpjudg{\process_{n-1}}{\NatPCF}\ar[d]\ar@{.}[r] &
          \tder_{n-1}\proves\ejudg{}{\I_{n-1}}{\R_{n-1}}{\NatU{m}}\ar@{-->}[u]
          \\
          \pcftder_n\proves\pcfpjudg{\process_{n}}{\NatPCF}\ar@{.}[r]&
          \tder_n\proves\ejudg{}{\I_n}{\process_{n}}{\NatU{m}}\ar[u]
        }
        $$
      \end{minipage}}
    }{
      $$\hspace{-5pt}
      \xymatrix{
        \pcftder_1\proves\pcfpjudg{\process_1}{\NatPCF}\ar[d]\ar@{.}[r] &
        \tder_1\proves\ejudg{}{\I_1}{\process_1}{\NatU{m}} \\
        \pcftder_2\proves\pcfpjudg{\process_2}{\NatPCF}
        \ar@{-->}[d]\ar@{.}[r] &
        \tder_2\proves\ejudg{}{\I_2}{\process_2}{\NatU{m}}\ar[u] \\
        \pcftder_{n-1}\proves\pcfpjudg{\process_{n-1}}{\NatPCF}\ar[d]\ar@{.}[r] &
        \tder_{n-1}\proves\ejudg{}{\I_{n-1}}{\R_{n-1}}{\NatU{m}}\ar@{-->}[u]
        \\
        \pcftder_n\proves\pcfpjudg{\process_{n}}{\NatPCF}\ar@{.}[r]&
        \tder_n\proves\ejudg{}{\I_n}{\process_{n}}{\NatU{m}}\ar[u]
      }
      $$
    }
  \caption{Completeness for Programs: sketch of the Proof}
  \label{fig:complprog}
\end{figure}                                               %
It should be now clear that one can reach the following:
\begin{theorem}[Completeness for Programs]
  \label{theo:cpl-prog}
  Suppose that $\pcftjudg{}{\t}{\NatPCF}$,
  that $\eval{\t}{n}{\nb[m]}$ and that $\ep$ is universal. 
  Then, $\ejudg{}{\inb[k]}{\t}{\NatU{\inb[m]}}$,
  where $k\leq n$.
\end{theorem}
\paragraph{Uniformisation and Completeness for Functions}
Completeness for programs, however, is not satisfactory: the fact
(normalising) \PCF\ terms of type $\NatPCF$ can all be analysed by 
\dlpcfv\ is not so surprising, and other type systems (like
non-idempotent intersection types~\cite{DeCarvalho}) have comparable expressive
power. Suppose we want to generalise relative completeness 
to first-order functions: we would like to prove that
every term $\t$ having a \PCF\ type $\NatPCF\arr\NatPCF$ 
(which terminates when fed with any natural number)
can be typed in \dlpcfv. How could we proceed?
First of all, observe that the argument in Figure~\ref{fig:complprog} 
could be applied to all \emph{instances} of $\t$, namely to all
terms in $\{\t n\mid n\in\int\}$. This way one can
obtain, for every $n\in\int$, a type derivation $\pcftder_n$ of
$$
\ejudg{}{\I_n}{\t}{\mtyp{\J_n}{\freccia{\NatU{\K_n}}{\NatU{\H_n}}}}
$$
where $\J_n$ can be assumed to be $1$, while $\K_n$ can be assumed to be
$n$. Moreover, the problem of obtaining $\pcftder_n$ from $n$ is recursive, i.e.,
can be solved by an algorithm. Surprisingly,
the infinitely many type derivations in $\{\pcftder_n\mid n\in\int\}$
can be turned into one:
\begin{proposition}[Uniformisation of type derivations]
  Suppose that $\ep$ is universal and that $\{\pcftder_n\}_{n\in\int}$ is a recursively enumerable class
  of type derivations satisfying the following constraints:
  \begin{varenumerate}
  \item
    For every $n\in\int$,
    $\tder_n\proves\ejudg{}{\I_n}{\t}{`s_n}$;
  \item
    all derivations have the same skeleton
  \end{varenumerate}
  Then there is a type derivation
  $\tder\proves\judg{a}{\emptyset}{\emptyset}{\I}{\t}{`s}$ such that 
  $\eijudg{\I\isubst{\inb}=\I_n}$
  and $\eijudg{`s\isubst{\inb}`=`s_n}$ for all~$n$.
\end{proposition}
Uniformisation of type derivations should be seen as an extreme form
of joining: not only a finite number of type derivations for the
same term can be unified into one, but even any recursively
enumerable class of them can. Again, the universality
of $\ep$ is crucial here. We are now ready to give the following:
\begin{theorem}[Completeness for functions]
  \label{theo:cpl-fct}
  Suppose that $\pcftjudg{}{\t}{\NatPCF\arr\NatPCF}$,
  that $\eval{\t\,\nb}{k_n}{\nb[m_n]}$ for all $n`:\int$
  and that $\ep$ is universal. 
  Then, there is an index~\H\ 
  such that
  $\judg{a}{\emptyset}{\emptyset}{\I}{\t}{\mtyp[b]{\one}{\freccia{\NatU{a}}{\NatU{\H}}}}$,
  where $\eijudg{\I\isubst{\inb}\leq \inb[k]_n}$ and
  $\eijudg{\H\isubst{\inb}=\inb[m]_n}$.
\end{theorem}


%% file: fig-stacktype.tex
\begin{math}
  \infer[]
  {\stjudg{\fiv}{\ictx}{\zero}{\emptystack}{`t}{`t}}
  {}
  \qquad\qquad
  \infer[]
  {\stjudg{\fiv}{\ictx}{\J}{\stack}{`s'}{`t'}}
  {
    \stjudg{\fiv}{\ictx}{\I}{\stack}{`s}{`t}\quad
    \sjudg{\fiv}{\ictx}{`s'`<`s}\quad
    \sjudg{\fiv}{\ictx}{`t`<`t'}\quad
    \ijudg{\fiv}{\ictx}{\I\leq\J}
  }
\end{math}

\vspace{10pt}
\begin{math}
  \infer[]
  {\stjudg{\fiv}{\ictx}{\J+\K}{\stackarg[\stack']{\cloc}}{\mtyp{\one}{(\freccia{`s}{`t})}}{`t'}}
  {
      \cjudg{\fiv}{\ictx}{\J}{\cloc}{`s\isubst{\zero}}\qquad
      \stjudg{\fiv}{\ictx}{\K}{\stack'}{`t\isubst{\zero}}{`t'}
  }
  \qquad\qquad
  \infer[]
  {\stjudg{\fiv}{\ictx}{\J+\K}{\stackfun[\stack']{\clov}}{`s\isubst{\zero}}{`t'}}
  {
      \cjudg{\fiv}{\ictx}{\J}{\clov}{\mtyp{\one}{(\freccia{`s}{`t})}}\qquad
      \stjudg{\fiv}{\ictx}{\K}{\stack'}{`t\isubst{\zero}}{`t'}
  }
\end{math}

\vspace{10pt}
\begin{math}
  \infer[]
  {\stjudg{\fiv}{\ictx}{\J+\K}{\stackif[\stack']{\t}{\u}{\env}}{\Nat{\M}{\N}}{`t}}
  {
      \cjudg{\fiv}{\N=\zero,\ictx}{\J}{\clo{\t}}{`s}\quad
      \cjudg{\fiv}{\M\geq\one,\ictx}{\J}{\clo{\u}}{`s}\quad
      \stjudg{\fiv}{\ictx}{\K}{\stack'}{`s}{`t}
  }
\end{math}

\vspace{10pt}
\begin{math}
  \infer[]
  {\stjudg{\fiv}{\ictx}{\I}{\stacks[\stack]}{\Nat{\M}{\N}}{`t}}
  {
    \stjudg{\fiv}{\ictx}{\I}{\stack}{\Nat{\M+1}{\N+1}}{`t}
  }
  \qquad\qquad
  \infer[]
  {\stjudg{\fiv}{\ictx}{\I}{\stackp[\stack]}{\Nat{\M}{\N}}{`t}}
  {
    \stjudg{\fiv}{\ictx}{\I}{\stack}{\Nat{\M-1}{\N-1}}{`t}
  }
\end{math}


%% file: fig-stacksize.tex
\begin{minipage}[t]{.4\linewidth}
  \textbf{Syntactic size of terms:}
  \begin{align*}
    \ts{\nb}&=2\\
    \ts{`lx.\t}&=\ts{\t}+2\\
    \ts{\fix{t}}&=\ts{\t}+2\\
    \ts{x}&=2\\
    \ts{\t\u}&=\ts{\t}+\ts{\u}+2\\
    \ts{\suc(\t)}&=\ts{\t}+2\\
    \ts{\pred(\t)}&=\ts{\t}+2\\
    \ts{\ifz{\t}{\u}{\s}}&=\ts{\t}+\ts{\u}+\ts{\s}+2
  \end{align*}
\end{minipage}
\begin{minipage}[t]{.5\linewidth}
  \textbf{Size of closures:}\quad
  $\sic{\clo{\t}}~=~\ms{\t}$ \\

  \textbf{Size of processes:}\quad
  $\sip{\proc{\cloc}{\stack}}~=~\sic{\cloc}+\sis{\stack}$ \\

  \textbf{Size of stacks:}
  \vspace{-1em}
  \begin{align*}
    \sis{\emptystack}~=~&0\\
    \sis{\stackfun{\clov}}~=~&\sic{\clov}+\sis{\stack}\\
    \sis{\stackarg{\cloc}}~=~&\sic{\cloc}+\sis{\stack}+1\\
    \sis{\stackif{\t}{\u}{\env}}~=~&\ms{\t}+\ms{\u}+\sis{\stack}+1\\
    \sis{\stacks}~=~&\sis{\stack}+1\\
    \sis{\stackp}~=~&\sis{\stack}+1\\
  \end{align*}
\end{minipage}


%% file: dev.tex
\section{Further Developments}\label{sec:dev}
Relative completeness of \dlpcfv, especially in its stronger form (Theorem~\ref{theo:cpl-fct}) can be read as follows.
Suppose that a (sound), finitary formal $\mathsf{C}$ system deriving judgements in the form $\sjudg{\fiv}{\ictx}{\I\leq\J}$ is fixed and ``plugged'' into \dlpcfv.
What you obtain is a sound, but necessarily incomplete formal system, due to G\"odel's incompleteness. However, this incompleteness is \emph{only} due to $\mathsf{C}$ and not to the rules of \dlpcfv, which are designed so as to reduce the problem of proving properties
of programs to checking inequalities over $\ep$ \emph{without any loss of information}.

In this scenario, it is of paramount importance to devise techniques to \emph{automatically} reduce the problem
of checking whether a program satisfies a given intentional or extensional specification to the
problem of checking whether a given set of inequalities over an equational program $\ep$ hold. Indeed, many
techniques and concrete tools are available for the latter problem (take, as an example,
the immense literature on SMT solving), while the same cannot be said about the former problem.
The situation, in a sense, is similar to the one in the realm of program logics for imperative programs, where
logics are indeed very powerful~\cite{Cook78}, and great effort have been directed to devise efficient
algorithms generating weakest preconditions~\cite{deBakker}.

Actually, at the time of writing, the authors are actively involved in the development of  
\emph{relative type inference} algorithms for both \dlpcfn\ and \dlpcfv, which can be
seen as having the same role as algorithms computing weakest preconditions. This is however out of the scope
of this paper.


%% file: concl.tex
\section{Conclusions}\label{sec:concl}
Linear dependent types are shown to be applicable to the analysis of intentional and
extensional properties of functional programs when the latter are call-by-value evaluated.
More specifically, soundness and relative completeness results are proved for
both programs a and functions. This generalises previous 
work by Gaboardi and the first author~\cite{DLG11}, who proved similar results in the call-by-name
setting. This shows that linear dependency not only provides an expressive formalism, but is
also robust enough to be adaptable to calculi whose notions of reduction are significantly
different (and more efficient) than normal order evaluation.

Topics for future work include some further analysis about the applicability of linear dependent
types to languages with more features, including some form of inductive data types, or
ground type references.
